\documentclass[a4paper,11pt]{article}
\usepackage[a4paper,margin=1in]{geometry}
\usepackage[utf8]{inputenc}
\usepackage{amsmath}
\usepackage{amssymb}
\usepackage{amsthm}
\usepackage{booktabs}
\usepackage[small]{caption}
\usepackage{cite}
\usepackage{colortbl}
\usepackage{enumitem}
\usepackage{framed}
\usepackage{graphicx}
\usepackage{microtype}
\usepackage{xcolor}
\usepackage{hyperref}
\newcommand{\lcllabel}{\Gamma}

\newcommand{\PP}{\Pi}
\newcommand{\N}{\mathbb{N}}
\newcommand{\A}{\mathcal{A}}
\newcommand{\C}{\mathcal{C}}
\newcommand{\CE}{\C_{\mathsf{edge}}}
\newcommand{\CV}{\C_{\mathsf{node}}}
\newcommand{\CH}{\C_{\mathsf{start}}}
\newcommand{\CT}{\C_{\mathsf{end}}}
\newcommand{\LL}{\mathcal{L}}
\newcommand{\M}{\mathcal{M}}
\newcommand{\GG}{\mathcal{G}}
\newcommand{\LCL}{\mathsf{LCL}}
\newcommand{\LOCAL}{\mathsf{LOCAL}}
\newcommand{\CONGEST}{\mathsf{CONGEST}}
\newcommand{\PRAM}{\mathsf{PRAM}}

\renewcommand{\P}{\mathsf{P}}
\newcommand{\PSPACE}{\mathsf{PSPACE}}
\newcommand{\EXPTIME}{\mathsf{EXPTIME}}
\newcommand{\NP}{\mathsf{NP}}
\newcommand{\coNP}{\mbox{co-$\NP$}}
\DeclareMathOperator{\poly}{poly}
\newcommand{\Ptwocol}{\mathsf{2col}}
\newcommand{\Pthreecol}{\mathsf{3col}}
\newcommand{\Porient}{\mathsf{orient}}

\newcommand{\unary}{o}
\newcommand{\tLCL}{\texorpdfstring{\boldmath$\LCL$}{LCL}}

\newtheorem{theorem}{Theorem}[section]
\newtheorem{lemma}[theorem]{Lemma}
\newtheorem{corollary}[theorem]{Corollary}

\theoremstyle{definition}
\newtheorem{definition}[theorem]{Definition}
\newtheorem{example}[theorem]{Example}
\theoremstyle{remark}
\newtheorem*{remark}{Remark}

\definecolor{myblue}{HTML}{0088cc}
\definecolor{myorange}{HTML}{f26924}

\hypersetup{
    colorlinks=true,
    linkcolor=black,
    citecolor=black,
    filecolor=black,
    urlcolor=myblue,
    pdftitle={Distributed graph problems through an automata-theoretic lens},
    pdfauthor={Yi-Jun Chang, Jan Studen\'y, Jukka Suomela}
}

\newcommand{\myemail}[1]{\,$\cdot$\, {\small #1}}
\newcommand{\myaff}[1]{\,$\cdot$\, {\small #1}\par\smallskip}

\newenvironment{myabstract}
{\list{}{\listparindent 1.5em
        \itemindent    \listparindent
        \leftmargin    0cm
        \rightmargin   0cm
        \parsep        0pt}%
    \item\relax}
{\endlist}

\newenvironment{mycover}
{\list{}{\listparindent 0pt
        \itemindent    \listparindent
        \leftmargin    0cm
        \rightmargin   0cm
        \parsep        0pt}%
    \raggedright
    \item\relax}
{\endlist}

\begin{document}

\begin{mycover}
{\huge\bfseries\boldmath Distributed graph problems through an~automata-theoretic lens \par}
\bigskip
\bigskip

\textbf{Yi-Jun Chang}
\myemail{cyijun@nus.edu.sg}
\myaff{National University of Singapore}

\textbf{Jan Studen\'y}
\myemail{jan.studeny@aalto.fi}
\myaff{Aalto University}

\textbf{Jukka Suomela}
\myemail{jukka.suomela@aalto.fi}
\myaff{Aalto University}
\bigskip
\end{mycover}

\begin{myabstract}
\noindent\textbf{Abstract.}
The \emph{locality} of a graph problem is the smallest distance $T$ such that each node can choose its own part of the solution based on its radius-$T$ neighborhood. In many settings, a graph problem can be solved efficiently with a distributed or parallel algorithm if and only if it has a small locality.

In this work we seek to \emph{automate} the study of solvability and locality: given the description of a graph problem $\Pi$, we would like to determine if $\Pi$ is solvable and what is the asymptotic locality of $\Pi$ as a function of the size of the graph. Put otherwise, we seek to automatically \emph{synthesize} efficient distributed and parallel algorithms for solving $\Pi$.

We focus on \emph{locally checkable} graph problems; these are problems in which a solution is globally feasible if it looks feasible in all constant-radius neighborhoods. Prior work on such problems has brought primarily bad news: questions related to locality are undecidable in general, and even if we focus on the case of \emph{labeled} paths and cycles, determining locality is $\PSPACE$-hard (Balliu et al., PODC 2019).

We complement prior negative results with efficient algorithms for the cases of \emph{unlabeled} paths and cycles and, as an extension, for rooted trees. We study locally checkable graph problems from an automata-theoretic perspective by representing a locally checkable problem $\Pi$ as a \emph{nondeterministic finite automaton} $\M$ over a \emph{unary alphabet}. We identify polynomial-time-computable properties of the automaton $\M$ that near-completely capture the solvability and locality of $\Pi$ in cycles and paths, with the exception of one specific case that is $\coNP$-complete.
\end{myabstract}
\medskip

\section{Introduction}\label{sec:intro}

In this work, our goal is to \emph{automate} the design of efficient distributed and parallel algorithms for solving graph problems, as far as possible. In the full generality, such tasks are undecidable: for example, given a Turing machine $M$, we can easily construct a graph problem $\Pi$ such that there is an efficient distributed algorithm for solving $\Pi$ if and only if $M$ halts~\cite{Naor1995}. Nevertheless, we are bringing here good news.

We focus on so-called \emph{locally checkable} graph problems in \emph{paths, cycles, and rooted trees}, and we show that in many cases, the task of designing efficient distributed or parallel algorithms for such problems can be automated, not only in principle but also in practice. 

We study the \emph{locality} of graph problems from an \emph{automata-theoretic perspective}. To introduce the concrete research questions that we study, we first define one specific model of distributed computing, the $\LOCAL$ model---through this model we can define the fundamental concept of locality. However, as we will later see, our results are directly applicable in many other synchronous models of distributed and parallel computing as well.

\paragraph{Background: locality and round complexity in distributed computing.}

In classical centralized sequential computing, a particularly successful idea has been the comparison of deterministic and nondeterministic models of computing. The question of $\P$ vs.\ $\NP$ is a prime example: given a problem in which solutions are easy to verify, is it also easy to solve?

In distributed computing a key computational resource is locality, and hence the distributed analogue of this idea can be phrased as follows: given a problem in which solutions can be verified locally, can it also be solved locally?

This question is formalized in the study of so-called \emph{locally checkable labeling} ($\LCL$) problems in the $\LOCAL$ model of distributed computing. $\LCL$ problems are graph problems in which solutions are \emph{labelings} of nodes and/or edges that can be \emph{verified locally}: if a solution looks feasible in all constant-radius neighborhoods, then it is also globally feasible \cite{Naor1995}. A simple example of an $\LCL$ problem is proper $3$-coloring of a graph: if a labeling of the nodes looks like a proper $3$-coloring in the radius-$1$ neighborhood of each node, then it is by definition a feasible solution.

In the $\LOCAL$ model of computing~\cite{Linial1992,Peleg2000}, we assume that the nodes of the input graph are equipped with unique identifiers from $\{1,2,\dotsc,\poly(n)\}$, where $n$ is the number of nodes. A distributed algorithm with a time complexity $T(n)$ is then a function that maps the radius-$T(n)$ neighborhood of each node into its local output. The local output of a node is its own part of the solution, e.g., its own color in the graph coloring problem. Here we say that the algorithm has locality $T$; the locality of a problem is the smallest $T$ such that there is an algorithm for solving it with locality $T$.

If we interpret the input graph as a computer network, with nodes as computers and edges as communication links, then in $T$ synchronous communication rounds all nodes can gather full information about their radius-$T$ neighborhood. Hence time (number of communication rounds) and distance (how far one needs to see) are interchangeable in the $\LOCAL$ model. In what follows, we will primarily use the term \emph{round complexity}.

\paragraph{Prior work: the complexity landscape of \tLCL{} problems.}

Now we have a natural distributed analog of the classical $\P$ vs.\ $\NP$ question: given an $\LCL$ problem, what is its round complexity in the $\LOCAL$ model? This is a question that was already introduced by Naor and Stockmeyer in 1995~\cite{Naor1995}, but the systematic study of the complexity landscape of $\LCL$ questions was started only very recently, around 2016~\cite{Brandt2016,Balliu2018stoc,Balliu2018disc,Balliu2019padding,Chang2019,Ghaffari2018,fischer17sublogarithmic,chang16exponential,ghaffari17distributed,Rozhon2020,Brandt2017}.

By now we have got a relatively complete understanding of \emph{possible complexity classes}: to give a simple example, if we look at deterministic algorithms in the $\LOCAL$ model, there are $\LCL$ problems with complexity $\Theta(\log^* n)$, and there are also $\LCL$ problems with complexity $\Theta(\log n)$, but it can be shown that there is no $\LCL$ problem with complexity between $\omega(\log^* n)$ and $o(\log n)$~\cite{Brandt2016,chang16exponential}.

However, much less is known about \emph{how to decide the complexity} of a given $\LCL$ problem. Many such questions are undecidable in general, and undecidability holds already in relatively simple settings such as $\LCL$s on 2-dimensional grids and tori~\cite{Brandt2017,Naor1995}. We will zoom into graph classes in which no such obstacle exists.

\paragraph{Our focus: cycles, paths, and rooted trees.}

Throughout this work, our main focus will be on paths and cycles. This may at first seem highly restrictive, but as we will show in Section~\ref{sec:trees}, once we understand $\LCL$ problems in paths and cycles, through reductions we will also gain understanding on so-called edge-checkable problems in rooted trees.

In cycles and paths, there are only three possible round complexities: $O(1)$, $\Theta(\log^* n)$, or $\Theta(n)$~\cite{Balliu2019decidable}. Randomness does not help in cycles and paths---this is a major difference in comparison with trees, in which there are $\LCL$ problems in which randomness helps exponentially~\cite{BEPS16,chang16exponential,PettieS-trianglefree}.

If our input is a \emph{labeled} path or cycle, the round complexity is known to be decidable, but unfortunately it is at least $\PSPACE$-hard~\cite{Balliu2019decidable}. On the other hand, the round complexity of $\LCL$s on \emph{unlabeled directed cycles} has a simple graph-theoretic characterization~\cite{Brandt2017}.

However, many questions are left open by prior work, and these are the questions that we will resolve in this work:
\begin{itemize}[noitemsep]
    \item What happens in \emph{undirected} cycles?
    \item What happens if we study \emph{paths} instead of cycles?
    \item Can we also characterize the \emph{existence} of a solution for all graphs in a graph class?
\end{itemize}
To illustrate these questions, consider the following problems that can be expressed as $\LCL$s:
\begin{itemize}[noitemsep]
    \item $\Pi_\Ptwocol$: finding a proper $2$-coloring,
    \item $\Pi_\Porient$: finding a globally consistent orientation (i.e., an orientation of edges such that it does not contain a node with two incoming or outgoing edges).
\end{itemize}
The round complexity of $\Pi_\Ptwocol$ is $\Theta(n)$ both in cycles and paths, regardless of whether they are directed or undirected, while the complexity of $\Pi_\Porient$ is $\Theta(n)$ in the undirected setting but it becomes $O(1)$ in the directed setting. Problems $\Pi_\Ptwocol$ and $\Pi_\Porient$ are always solvable on paths, and $\Pi_\Porient$ is always solvable on cycles, but if we have an odd cycle, then a solution to $\Pi_\Ptwocol$ does not exist. In particular, for $\Pi_\Ptwocol$ there are infinitely many solvable instances and infinitely many unsolvable instances. Our goal in this work is to develop a framework that enables us to make these kind of observations \emph{automatically} for any given $\LCL$ problem.

\paragraph{\tLCL{}s as nondeterministic automata over a unary alphabet.}

In this work we study the   solvability and the round complexity of $\LCL$ problems from an automata-theoretic perspective.
Specifically, we generalize the graph-theoretic characterization for $\LCL$ problems on 
 unlabeled directed cycles in~\cite{Brandt2017} to all paths and cycles, directed and undirected, and  identify a connection between such a characterization and automata theory. 
 
 This connection allows us to leverage prior work on automata theory.  For example, as we will later see in this work, the co-$\NP$-completeness of the universality problem for  nondeterministic finite automata~\cite{Stockmeyer73} allows us to deduce the $\NP$-hardness for distinguishing between zero and infinitely many unsolvable instances for $\LCL$ problems on paths.

We would like to emphasize that there are many ways to interpret $\LCL$s as automata---and the approach that might seem most natural does not make it possible to directly leverage prior work on automata theory. 
We will later see that the approach we take enables us to identify direct connections between distributed computational complexity and automata theory.

Let us first briefly describe the ``obvious'' encoding and show why it does not achieve what we want: A labeling of a directed path with symbols from some alphabet $\Sigma$ can be interpreted as a string. Then a locally checkable problem can be interpreted as a regular language over alphabet $\Sigma$. We can then represent an $\LCL$ problem $\Pi$ as a finite automaton $\M$ such that $\M$ accepts a string $x \in \Sigma^\ast$ if and only if a directed path labeled with $x$ is a feasible solution to~$\Pi$.

However, such an interpretation does \emph{not} seem to lead to a useful theory of $\LCL$ problems. To see one challenge, consider these problems on paths:
\begin{itemize}[noitemsep]
    \item $\Pi_\Ptwocol$: finding a proper $2$-coloring,
    \item $\Pi_\Pthreecol$: finding a proper $3$-coloring.
\end{itemize}
These are fundamentally different problems from the perspective of $\LCL$s in the $\LOCAL$ model: problem $\Pi_\Ptwocol$ requires $\Theta(n)$ rounds while problem $\Pi_\Pthreecol$ is solvable in $\Theta(\log^* n)$ rounds~\cite{cole86deterministic}. However, if we consider analogous automata $\M_\Ptwocol$ and $\M_\Pthreecol$ that recognize these solutions, it is not easy to identify a classical automata-theoretic concept that would separate these cases.

Instead of identifying the \emph{alphabet} of the automaton with the set of labels in the $\LCL$, it turns out to be a better idea to have a \emph{unary} alphabet and identify the \emph{set of states} of the automaton with the set of labels. In brief, the perspective that we take throughout this work is as follows (this is a simplified version of the idea):
\begin{framed}
\noindent
Assume $\Pi$ is an $\LCL$ problem in which the set of output labels is $\Gamma$. We interpret $\Pi$ as a \emph{nondeterministic finite automaton} $\M_\Pi$ over the \emph{unary alphabet} $\Sigma = \{ \unary \}$ such that the set of states of $\M_\Pi$ is $\Gamma$.
\end{framed}
\noindent At first this approach may seem counter-intuitive. But as we will see in this work, it enables us to connect classical automata-theoretic concepts to properties of $\LCL$s this way.

To give one nontrivial example, consider the question of whether a given $\LCL$ problem $\Pi$ can be solved in $O(\log^* n)$ rounds. With the above interpretation, this turns out to be directly connected to the existence of \emph{synchronizing words}~\cite{Cerny1964,eppstein1990reset}, in the following nondeterministic sense: we say that $w$ is a synchronizing word for an NFA $\M$ that takes $\M$ into state $t$ if, given any starting state $s \in Q$ there is a sequence of state transitions that takes $\M$ to state $t$ when it processes~$w$. Such a sequence $w$ is known as the \emph{D3-directing word} introduced in~\cite{Imreh1999} and further studied in~\cite{don2018synchronizing,imreh2005regular,GAZDAG2009986,Martyugin2014}. We will show that the following holds (up to some minor technicalities):
\begin{framed}
\noindent An $\LCL$ on directed paths and cycles has a round complexity of $O(\log^* n)$ if and only if a strongly connected component of the corresponding NFA $\M$ over the unary alphabet has a D3-directing word.
\end{framed}
\noindent Moreover, we will show that for the unary alphabet, the existence of such a word can be decided in \emph{polynomial time} in the size of the NFA $\M$, or equivalently, in the size of the description of the $\LCL$ $\Pi$. In contrast, when the size of the alphabet is at least two, the problem of deciding the existence of a D3-directing word is known to be $\PSPACE$-hard~\cite{Martyugin2014}.

We would like to emphasize that this connection between $\LCL$ problems and automata theory is not inherent to unlabeled paths and cycles.  For example, \emph{tree automata} can be used to encode   $\LCL$ problems on bounded-degree trees, and to encode $\LCL$ problems with input labels $\Sigma$, it suffices to consider  automata over the alphabet $\Sigma$. Whether such a  connection  beyond unlabeled paths and cycles can lead to new results is an interesting future work direction.

\paragraph{Contributions.}

We study $\LCL$ problems in unlabeled cycles and paths, both with and without consistent orientation. For each of these settings, we show how to answer the following questions in a mechanical manner, for any given $\LCL$ problem $\Pi$:
\begin{itemize}[noitemsep]
    \item How many unsolvable instances there are (none, finitely, or infinitely many)?
    \item How many solvable instances there are (none, finitely, or infinitely many)?
    \item What is the round complexity of $\Pi$ for solvable instances ($O(1)$, $\Theta(\log^* n)$, or $\Theta(n)$)?
\end{itemize}
We show that all such questions are not only decidable but they are in $\NP$ or $\coNP$, and almost all such questions are in $\P$, with the exception of a couple of specific questions that are $\NP$-complete or $\coNP$-complete. We also give a complete classification of all possible case combinations---for example, we show that if there are infinitely many unsolvable instances, then the complexity of the problem for solvable instances cannot be $\Theta(\log^* n)$.

We give a uniform automata-theoretic formalism that enables us to study such questions, and that makes it possible to leverage prior work on automata theory. We also develop new efficient algorithms for some automata-theoretic questions that to our knowledge have not been studied before.

Finally, we show that our results can be used to analyze also a family of $\LCL$ problems in rooted trees. This demonstrates that the automata-theoretic framework considered here is applicable also beyond the seemingly restrictive case of cycles and paths.

Our main result---the complete classification of the solvability and distributed round complexity of all $\LCL$ problems in undirected and directed cycles and paths is presented in Table~\ref{tab:cycles}.

\begin{table}[t!]
\newcommand{\pcol}[1]{\textcolor{myblue}{#1}}
\newcommand{\ncol}[1]{\textcolor{myorange}{#1}}
\newcommand{\yes}{\pcol{yes}}
\newcommand{\no}{\ncol{no}}
\newcommand{\ifapplicable}{$^\dagger$}

\newcommand{\const}{$\square$}
\newcommand{\logstar}{$\boxplus$}
\newcommand{\linear}{$\blacksquare$}

\newcommand{\xzro}{$0$}
\newcommand{\xcon}{$<$}
\newcommand{\xfty}{$\infty$}
\newcommand{\xhard}{?}
\newcommand{\goodzro}{{\xzro}}
\newcommand{\goodcon}{{{\xcon}}}
\newcommand{\goodfty}{{\xfty}}
\newcommand{\badzro}{{\xzro}}
\newcommand{\badcon}{{{\xcon}}}
\newcommand{\badfty}{{\xfty}}
\newcommand{\badhard}{{\xhard}}
\newcommand{\dontcare}{---}
\newcommand{\na}{$\times$}
\newcommand{\colsp}{\hspace{4.5mm}}
\newcommand{\lfsep}{\hspace{3.5mm}}
\newcommand{\lsep}{\hspace{3.5mm}}
\centering
\begin{tabular}{@{}l@{\colsp\colsp}c@{\colsp}c@{\colsp}c@{\colsp}c@{\colsp}c@{\colsp}c@{\colsp}c@{\colsp}c@{\colsp}c@{\colsp}c@{\colsp}c@{}}
\toprule
Type
& \textbf{A} & \textbf{B} & \textbf{C} & \textbf{D} & \textbf{E} & \textbf{F} & \textbf{G} & \textbf{H} & \textbf{I} & \textbf{J} & \textbf{K} \\
\midrule
Def.~\ref{def:symmetric}: symmetric problem& \yes   & \yes   & \yes&\no  & \yes   & \yes&\no  & \yes&\no  & \yes&\no \\
Def.~\ref{def:repeatable-state}: repeatable state            & \yes   & \yes   & \yes & \yes    & \yes   & \yes & \yes     & \yes & \yes    & \no & \no    \\
Def.~\ref{def:flexible-state}: flexible state \cite{Brandt2017}             & \yes   & \yes   & \yes & \yes     & \yes   & \yes & \yes     & \no & \no     & \no & \no    \\
Def.~\ref{def:loop}: loop \cite{Brandt2017}                       & \yes   & \yes   & \yes & \yes     & \no    & \no & \no      & \no & \no     & \no & \no    \\
Def.~\ref{def:mirror-flexible-state}: mirror-flexible state       & \yes   & \yes   & \no & \dontcare     & \yes   & \no & \dontcare     & \no & \dontcare     & \no & \dontcare    \\
Def.~\ref{def:mirror-flexible-loop}: mirror-flexible loop        & \yes   & \no    & \no & \dontcare     & \no    & \no & \dontcare     & \no & \dontcare     & \no & \dontcare    \\
\midrule
\addlinespace[6pt]
\multicolumn{12}{@{}l@{}}{\textbf{Number of instances:}\lfsep\xzro{} = zero\lsep\xcon{} = finite\lsep\xfty{} = infinite\lsep\xhard{} = $\NP$-complete to decide}\\
\addlinespace[5pt]
$\cdot$ solvable cycles     &\goodfty&\goodfty&\goodfty &\goodfty &\goodfty&\goodfty &\goodfty &\goodfty &\goodfty &\badzro &\badzro \\
$\cdot$ solvable paths      &\goodfty&\goodfty&\goodfty &\goodfty &\goodfty&\goodfty &\goodfty &\goodfty &\goodfty &\badcon &\badcon \\
$\cdot$ unsolvable cycles   &\goodzro&\goodzro&\goodzro &\goodzro &\goodcon&\goodcon &\goodcon &\badfty  &\badfty  &\badfty &\badfty \\
$\cdot$ unsolvable paths    &\goodcon&\goodcon&\goodcon &\goodcon &\goodcon&\goodcon &\goodcon &\badhard &\badhard &\badfty &\badfty \\
\midrule
\addlinespace[6pt]
\multicolumn{12}{@{}l@{}}{\textbf{Distributed round complexity:}\lfsep\const{} = $O(1)$\lsep\logstar{} = $\Theta(\log^* n)$\lsep\linear{} = $\Theta(n)$\lsep\na{} = N/A}\\
\addlinespace[5pt]
$\cdot$ directed cycles \cite{Brandt2017}    &\const  &\const  &\const &\const   &\logstar&\logstar &\logstar &\linear  &\linear  &\na      &\na     \\
$\cdot$ directed paths                       &\const  &\const  &\const &\const   &\logstar&\logstar &\logstar &\linear  &\linear  &\const   &\const  \\
$\cdot$ undirected cycles                    &\const  &\logstar&\linear&\na      &\logstar&\linear  &\na      &\linear  &\na      &\na      &\na     \\
$\cdot$ undirected paths                     &\const  &\logstar&\linear&\na      &\logstar&\linear  &\na      &\linear  &\na      &\const   &\na     \\
\bottomrule
\end{tabular}
\caption{Classification of $\LCL$ problems in cycles and paths. This table defines 11 types, labeled with A--K, based on six properties (Definitions \ref{def:symmetric}, \ref{def:repeatable-state}--\ref{def:mirror-flexible-loop}); see Figure~\ref{fig:cycles} for examples of problems of each type. For each problem type, we show what is the number of solvable instances, the number of unsolvable instances, and the distributed round complexity for both directed and undirected paths and cycles. The cases marked with ``\na'' refer to problems that are not well-defined or that are never solvable. For the cases labeled with ``\badhard'' deciding the number of unsolvable instances is $\NP$-complete (or $\coNP$-complete depending on the way one defines the decision problem); see Section~\ref{sec:paths}. However, for all other cases the type directly determines both solvability, and all these cases are also decidable in polynomial time; see Section~\ref{sec:efficient}. The correctness of this classification is proved in Section~\ref{sec:cycles-correct}.}
\label{tab:cycles}
\end{table}

\paragraph{Extensions to other models of distributed and parallel computing.}

While we use the $\LOCAL$ model of distributed computing throughout this work, our results are directly applicable also in many other models of distributed and parallel computing.

In distributed computing we usually assume that the input graph represents the communication network; each node is a computer, each edge is a communication link, and the nodes can communicate by passing messages to each other. However, in parallel computing we usually take a very different perspective: we assume that the input graph is stored as a linked data structure somewhere in the shared memory, and we have multiple processors that can access the memory. In such a setting, directed paths and rooted trees are particularly relevant families of input, as they correspond to linked lists and tree data structures.

While the settings are superficially different, our \emph{upper bounds} apply directly in all such settings. All of our distributed algorithms are based on the observation that there are two canonical problems: \emph{distance-$k$ anchoring} (Definition~\ref{def:anchoring}) and \emph{distance-$k$ orientation} (Definition~\ref{def:orientation}). Both of the canonical problems can be solved in the message-passing setting with small messages and with little local memory. Furthermore, when we look at rooted trees (Section~\ref{sec:trees}), our algorithms are ``one-sided'': each node only needs to receive information from its parent. It follows that our algorithms work also e.g.\ in the $\CONGEST$ model \cite{Peleg2000} of distributed computing, and they can be efficiently simulated e.g.\ in the classic $\PRAM$ model, as well as various modern models of massively parallel computing.

Our \emph{lower bounds} are also broadly applicable, as they hold in the $\LOCAL$ model, which is a very strong model of distributed computing (unbounded message size; unlimited local storage; unbounded local computation; nodes can talk to all of their neighbors in parallel). In particular, the lower bounds trivially hold also the $\CONGEST$ model. Adapting the lower bounds to shared-memory models takes more effort, but it is also possible---see Fich and Ramachandran \cite{Fich1990} for an example of how to turn $\Omega(\log^* n)$ lower bounds for the $\LOCAL$ model into $\Omega(\log \log^* n)$ lower bounds for variants of the $\PRAM$ model.

\paragraph{Comparison with prior work.}

In comparison with \cite{Balliu2019decidable,Brandt2017,chang16exponential,Chang2019,Naor1995}, our work gives a more fine-grained perspective: instead of merely discussing decidability, we explore the question of which of the decision problems are in $\P$, $\NP$, and $\coNP$.

In comparison with the discussion of directed cycles in \cite{Brandt2017}, our work studies a much broader range of settings. Previously, it was not expected that the simple characterization of $\LCL$s on directed cycles could be extended in a straightforward manner to paths or undirected cycles. For example, we can define an infinite family of orientation problems that can be solved in undirected cycles in $O(1)$ rounds but that require a nontrivial algorithm; such problems do not exist in directed cycles, as $O(1)$-round solvability implies trivial $0$-round solvability.

Furthermore, we study the graph-theoretic question of the existence of a solution in addition to the algorithmic question of the complexity of finding a solution, and relate solvability with complexity in a systematic manner; we are not aware of prior work that would do the same in the context of $\LCL$s in the $\LOCAL$ model.

Our work also takes the first steps towards an effective (i.e., polynomial-time computable)
characterization of $\LCL$ problems in trees, by showing how to characterize so-called edge-checkable problems in rooted trees. 

For general $\LCL$ problems on bounded-degree trees, previous work~\cite{Balliu2018disc,Chang2019,chang:LIPIcs:2020:13096} showed that it is decidable to distinguish between the complexity pairs $O(\log n)$ -- $n^{\Omega(1)}$  and $O(n^{1/(k+1)})$ -- $\Omega(n^{1/k})$ for any constant $k \geq 1$. These algorithms are not efficient, as these are $\EXPTIME$-hard problems~\cite{Chang2019}.

The previous work~\cite{Balliu2019decidable,Balliu2018disc,Chang2019,chang:LIPIcs:2020:13096} studying the complexity landscape of $\LCL$ problems on paths, cycles, and bounded-degree trees \emph{with input labels} uses a different connection to  automata theory. In their proofs, they classified paths and trees into a finite number of classes satisfying certain properties using the \emph{pumping lemma} for regular languages.

\paragraph{Organization.}
In Section~\ref{sec:prelim}, we formally define $\LCL$ problems and their representation as automata. In Sections~\ref{sec:cycles} and~\ref{sec:paths}, we present our classification of $\LCL$ problems on cycles and paths. In Section~\ref{sec:efficient}, we show that the classification is polynomial-time computable.
In Section~\ref{sec:cycles-correct}, we prove the correctness of our classification.
In Section~\ref{sec:trees}, we extend our classification to rooted trees.
In Section~\ref{sec:conclusion}, we give some concluding remark and point to some open problems.

\section{Representation of \tLCL{}s as automata}\label{sec:prelim}

To reiterate, $\LCL$ problems \cite{Naor1995}, broadly speaking, are problems in which the task is to label nodes and/or edges with labels from a constant-size alphabet (denoted by $\Gamma$), subject to local constraints. That is, a solution is globally feasible if it looks good in all radius-$r$ neighborhoods for some constant~$r$. In this section we will develop a way to represent all $\LCL$ problems on paths and cycles as  nondeterministic automata.

In this paper, we consider as input graphs only paths and cycles in which the edges are either undirected (undirected case) or consistently oriented (directed case). We say that a cycle or a path has \emph{consistently oriented edges} if it does not contain a node with two incoming or two outgoing edges.

\subsection{Formalizing \tLCL{}s as node-edge-checkable problems}

$\LCL$ problems can be specified in many different forms, and we have to be able to capture, among others, problems of the following forms:
\begin{itemize}[noitemsep]
    \item The problem may ask for a labeling of nodes, a labeling of edges, a labeling of the endpoints of the edges, an orientation of the edges, or any combination of these.
    \item The input graph can be a path or a cycle.
    \item The input graph may be directed or undirected.
\end{itemize}
As discussed in the recent papers \cite{Balliu2019a,Balliu2019}, a rather elegant way to capture all $\LCL$ problems is the following approach: We imagine that we have split every edge into two \emph{half-edges}, which are also called \emph{ports}. The labeling refers only to the ports.

More formally, a \emph{port} or a \emph{half-edge} $p$ is a pair $(e,v)$ consisting of an edge $e$ and a node $v \in e$ incident to~$e$. Let $P$ be the set of all ports. A \emph{labeling} is a function $\lambda\colon P \to \Gamma$ from ports to labels from some alphabet $\Gamma$.

It is easy to see that we can represent $\LCL$ problems of different flavors in this formalism, for example:
\begin{itemize}[noitemsep]
    \item If the task is to label nodes, we require all ports incident to a node to be labeled by the same label, so that the label of a node is well-defined.
    \item If the task is to label edges, we require that both half-edges of each edge have the same label, so that the label of an edge is well-defined.
    \item If the task is to find an orientation, we can use e.g.\ symbols $H$ (head) and $T$ (tail) and require that for each edge exactly one half is labeled with $H$ and the other half is labeled with $T$, so that the orientation of each edge is well-defined.
\end{itemize}

Moreover, the constraints for node-edge-checkable problems will be divided into \emph{node constraints} and \emph{edge constraints}. Node constraints consider only incident port labels of a node and edge constraints consider only incident port labels of an edge. 

We will now formally define an $\LCL$ problem in the node-edge-checkable formalism. Let us first consider the case of directed cycles or paths. By assumption, a directed cycle or a directed path is consistently oriented. For each edge, one port is a \emph{tail port} and the other port is a \emph{head port}. Furthermore, for each degree-$2$ node, there is also exactly one head port and exactly one tail port incident to it.
\begin{definition}[LCL problem]
An $\LCL$ problem $\Pi$ in the node-edge-checkable formalism on cycles or paths is a tuple $\Pi = (\Gamma,\CE,\CV,\CH,\CT)$ consisting of
\begin{itemize}[noitemsep]
    \item a finite set $\Gamma$ of output labels,
    \item an edge constraint $\CE \subseteq \Gamma \times \Gamma$,
    \item a node constraint $\CV \subseteq \Gamma \times \Gamma$, and
    \item start and end constraints $\CH \subseteq \Gamma$ and $\CT \subseteq \Gamma$.
\end{itemize}
\end{definition}

\begin{definition}[Solution on directed cycles or paths]
Let $G$ be a directed cycle or a directed path, and let $\Pi$ be an $\LCL$ problem, and let $\lambda\colon P \to \Gamma$ be a labeling of $G$. We say that $\lambda$ is a \emph{solution} to $\Pi$ if the following holds:
\begin{itemize}[noitemsep]
    \item For each edge $e$, if $p$ is the tail port and $q$ is the head port of $e$, then $(\lambda(p), \lambda(q)) \in \CE$.
    \item For each degree-$2$ node $v$, if $p$ is the head port and $q$ is the tail port of $v$, then $(\lambda(p), \lambda(q)) \in \CV$.
    \item For each degree-$1$ node $v$ with only one tail port $p$, we have $\lambda(p) \in \CH$.
    \item For each degree-$1$ node $v$ with only one head port $p$, we have $\lambda(p) \in \CT$.
\end{itemize}
\end{definition}

Informally, when we follow the labeling in the positive direction along the directed path, we will first see a label from $\CH$, then each edge is labeled with a pair from $\CE$, each internal node is labeled with a pair from $\CV$, and the final label along the path is $\CT$.

Next we consider the case of undirected cycles or paths. 

\begin{definition}[Symmetric LCL problems]\label{def:symmetric}
We say that an $\LCL$ problem $\Pi = (\Gamma,\CE,\allowbreak\CV,\allowbreak\CH,\CT)$ is \emph{symmetric} if $\CE$ and $\CV$ are symmetric relations and $\CH = \CT$. Otherwise the problem is \emph{asymmetric}.
\end{definition}

In the undirected case we cannot consistently distinguish ports, and hence we can only solve and define solution for symmetric LCL problems.

\begin{definition}[Solution on undirected cycles or paths]
Let $G$ be an undirected cycle or an undirected path, and let $\Pi$ be a symmetric $\LCL$ problem, and let $\lambda\colon P \to \Gamma$ be a labeling of $G$. We say that $\lambda$ is a \emph{solution} to $\Pi$ if the following holds:
\begin{itemize}[noitemsep]
    \item For each edge $e$, if the ports of $e$ are $p$ and $q$, then $(\lambda(p), \lambda(q)) \in \CE$. 
    \item For each degree-$2$ node $v$, if the ports incident to $v$ are $p$ and $q$, then $(\lambda(p), \lambda(q)) \in \CV$.
    \item For each degree-$1$ node $v$, if the port incident to $v$ is $p$, then $\lambda(p) \in \CH = \CT$.
\end{itemize}
\end{definition}

Recall that in symmetric problems $\CE$ and $\CV$ are symmetric, so the above formulation is well-defined. When we study the case of cycles, we can set $\CH = \CT = \emptyset$. For brevity, in what follows, we will usually write the pair $(a,b)$ simply as $ab$.

It is usually fairly easy to encode any given $\LCL$ problem in a natural manner in this formalism---see Figure~\ref{fig:ex-cycle} for examples. In the figure, \emph{maximal matching} serves as an example of a problem in which the natural encoding of indicating which edges are part of the matching does not work (it does not capture maximality) but with one additional label we can precisely define a problem that is equivalent to maximal matchings.

\begin{figure}[t!]
\centering
\includegraphics[page=1]{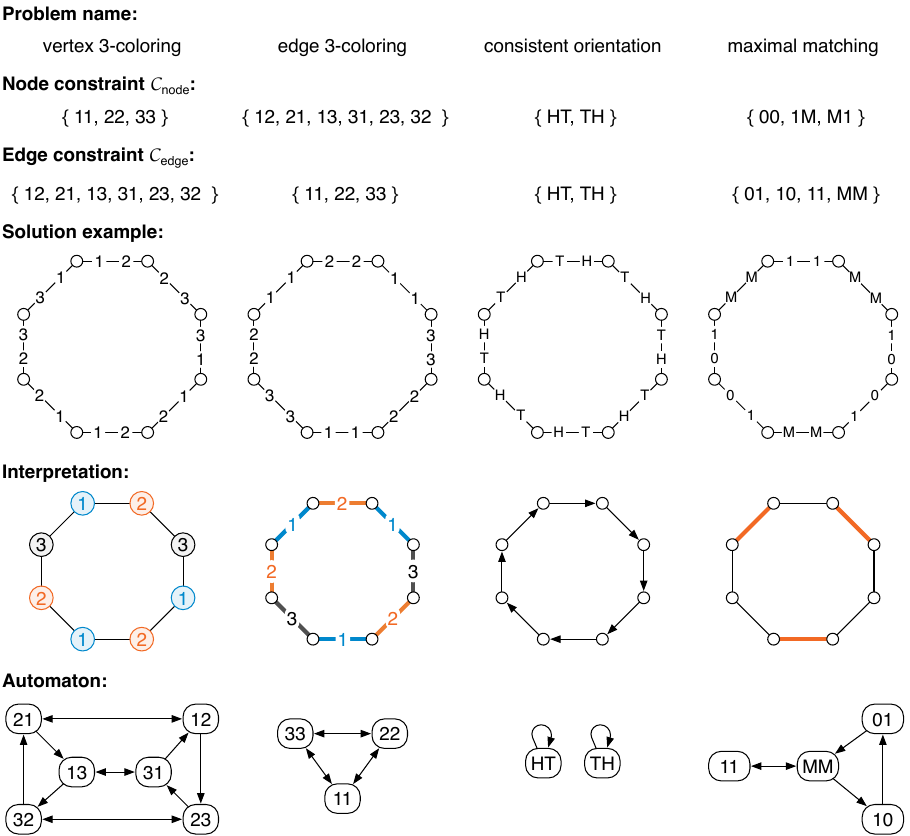}
\caption{Examples of how to encode $\LCL$ problems in the node-edge-checkable formalism, and how to represent the problem as an automaton. Here the problems are symmetric, so they are well-specified also on undirected cycles. For maximal matching, ports incident to matched nodes are labeled with ``1'' and ``M'', ports incident to unmatched nodes are labeled with ``0'', and the edge constraints ensure that there are no unmatched nodes adjacent to each other.}\label{fig:ex-cycle}
\end{figure}

In general, if we have any $\LCL$ problem $\Pi$ (in which the problem description can refer to radius-$r$ neighborhoods for some constant $r$), we can define an equivalent problem $\Pi'$ that can be represented in the node-edge-checkable formalism, modulo constant-time preprocessing and postprocessing. In brief, one label in the new problem $\Pi'$ corresponds to the labeling of a subpath of length $\Theta(r)$ in $\Pi$. Now given a solution of $\Pi$, one can construct a solution of $\Pi'$ in $O(r)$ rounds, and given a solution of $\Pi'$, one can construct a solution of $\Pi$ in zero rounds. Moreover, $\Pi'$ can be specified in the node-edge-checkable formalism. We will give the details in Section~\ref{sec:universality}. From now on, all $\LCL$ problems considered are by default problems defined using the node-edge-checkable formalism.

\begin{figure}
\centering
\includegraphics[page=2]{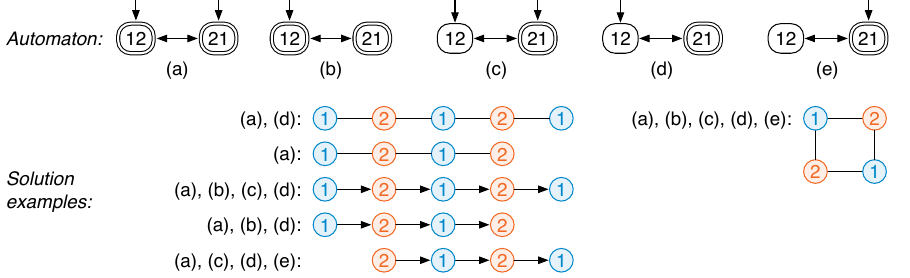}
\caption{Five variants of the node $2$-coloring problem. Each variant has different allowed colors for the endpoints, hence also different starting and accepting states. Here (a) and (d) are the only problems that are symmetric; therefore problems (b), (c), and (e) are not meaningful on undirected paths.}\label{fig:ex-path}
\end{figure}

\subsection{Turning node-edge-checkable problems into automata}

Now consider an $\LCL$ problem $\Pi$ that is specified in the node-edge-checkable formalism. Construct a nondeterministic finite automaton $\M_\Pi$ as follows; see Figures \ref{fig:ex-cycle} and \ref{fig:ex-path} for examples.
\begin{framed}
\begin{itemize}[nolistsep,leftmargin=*]
    \item The set of states is $\CE$.
    \item There is a transition from $(a,b)$ to $(c,d)$ whenever $(b,c) \in \CV$.
    \item $(a,b) \in \CE$ is a starting state whenever $a \in \CH$.
    \item $(a,b) \in \CE$ is an accepting state whenever $b \in \CT$.
\end{itemize}
\end{framed}
\noindent
We will interpret $\M_\Pi$ as an NFA over the unary alphabet $\Sigma = \{ \unary \}$. Note that there can be multiple starting states; the automaton can choose the starting state nondeterministically. We remark that in case of cycles, the sets $\CH$ and $\CT$ are empty which transforms an NFA into a nondeterministic semiautomaton (i.e., an automation having no starting or accepting states). In the following part we will see how to view the constructed automata. 

We define the following concepts:
\begin{definition}[generating paths and cycles]
Automaton $\M$ can \emph{generate} the cycle $(x_1, x_2,\allowbreak \dotsc,\allowbreak x_m)$ if each $x_i$ is a state of $\M$, there is a state transition from $x_i$ to $x_{i+1}$ for each $i < m$, and there is a state transition from $x_m$ to $x_1$.

Automaton $\M$ can \emph{generate} the path $(x_1, x_2, \dotsc, x_m)$ if each $x_i$ is a state of $\M$, $x_1$ is a starting state, $x_m$ is an accepting state, and there is a state transition from $x_i$ to $x_{i+1}$ for each $i < m$.
\end{definition}
Note that $\M$ can generate cycles even if there are no starting states or accepting states.

\begin{example}
\renewcommand{\TH}{\mathrm{TH}}
\newcommand{\HT}{\mathrm{HT}}
\newcommand{\MM}{\mathrm{MM}}
Consider the state machines in Figure~\ref{fig:ex-cycle}. The state machine for consistent orientation can generate the following cycles:
\begin{align*}
& (\HT),
\ (\TH),
\ (\HT,\, \HT),
\ (\TH,\, \TH),
\ (\HT,\, \HT,\, \HT),
\ (\TH,\, \TH,\, \TH),
\ \dotsc
\end{align*}
The state machine for maximal matching can generate the following cycles:
\begin{align*}
& (11,\, \MM),
\ (\MM,\, 11),
\ (10,\, 01,\, \MM),
\ (01,\, \MM,\, 10),
\ (\MM,\, 10,\, 01), \\
& (11,\, \MM,\, 11,\, \MM),
\ (\MM,\, 11,\, \MM,\, 11),
\ \dotsc
\end{align*}
\end{example}

\begin{remark}
If we start with a symmetric problem, the automaton will be \emph{mirror-symmetric} in the following sense: there is a state transition $(a,b) \to (c,d)$ if and only if there is a state transition $(d,c) \to (b,a)$, and the automaton can generate $(x_1 y_1, \dotsc, x_m y_m)$ if and only if it can generate $(y_m x_m, \dotsc, y_1 x_1)$. All automata in Figure~\ref{fig:ex-cycle} have this property, while in Figure~\ref{fig:ex-path} only automata (a) and (d) are mirror-symmetric.
\end{remark}

\paragraph{Automata capture node-edge-checkable problems.}

These observations follow directly from the definitions:
\begin{itemize}
    \item Let $\Pi$ be a symmetric or asymmetric problem. Automaton $\M_\Pi$ can generate a cycle $(x_1, x_2, \dotsc, x_m)$ if and only if the following is a feasible solution for problem $\Pi$: Take a \emph{directed} cycle with $m$ nodes and $m$ edges and walk along the cycle in the positive direction, starting at an arbitrary edge. Label the ports of the first edge with $x_1$, the ports of the second edge with $x_2$, etc.
    \item Let $\Pi$ be a symmetric problem. Automaton $\M_\Pi$ can generate a cycle $(x_1, x_2, \dotsc, x_m)$ if and only if the following is a feasible solution for problem $\Pi$: Take an \emph{undirected} cycle with $m$ nodes and $m$ edges and walk the cycle in some consistent direction, starting at an arbitrary edge. Label the ports of the first edge with $x_1$, the ports of the second edge with $x_2$, etc.
    \item Let $\Pi$ be a symmetric or asymmetric problem. Automaton $\M_\Pi$ can generate a path $(x_1, x_2, \dotsc, x_m)$ if and only if the following is a feasible solution for problem $\Pi$: Take a \emph{directed} path with $m+1$ nodes and $m$ edges and walk along the path in the positive direction, starting with the first edge. Label the ports of the first edge with $x_1$, the ports of the second edge with $x_2$, etc.
    \item Let $\Pi$ be a symmetric problem. Automaton $\M_\Pi$ can generate a path $(x_1, x_2, \dotsc, x_m)$ if and only if the following is a feasible solution for problem $\Pi$: Take an \emph{undirected} path with $m+1$ nodes and $m$ edges and walk along the path in some consistent direction, starting with the first edge. Label the ports of the first edge with $x_1$, the ports of the second edge with $x_2$, etc.
\end{itemize}
Hence, for example, the question of whether a given problem $\Pi$ is solvable in a path of length $m$ is equivalent to the question of whether $\M_\Pi$ accepts the string $\unary^m$. Similarly, the question of whether $\Pi$ is solvable in a cycle of length $m$ is equivalent to the question of whether there is a state $q$ such that $\M_\Pi$ can return to state $q$ after processing $\unary^m$.

However, the key question is what can be said about the complexity of solving $\Pi$ in a distributed setting. As we will see, this is also captured in the structural properties of $\M_\Pi$.

\subsection{Universality of the node-edge-checkable formalism}\label{sec:universality}

In this section, we show that the node-edge-checkable formalism is universal in the following sense. Let $\Pi$ be any $\LCL$ given in the standard format by listing all valid local neighborhoods of some constant radius $r$. We can construct an $\LCL$ problem $\Pi'$ that is in the node-edge-checkable formalism satisfying the following properties. 

\begin{description}
    \item[Efficiency:] Both the runtime of the construction and the description length of $\Pi'$ are polynomial in the description length of $\Pi$.
    \item[Equivalence:] Let the communication network $G$ be a cycle of length at least $2r+2$ or a path. Starting from any given legal labeling $\lambda$ for $\Pi$ on $G$, in $O(1)$ rounds we can transform it into a legal labeling $\lambda'$ for $\Pi'$. Similarly, starting from any given legal labeling $\lambda'$ for $\Pi'$ on $G$, in $O(1)$ rounds we can transform it into a legal labeling $\lambda$ for $\Pi$.
\end{description}

In particular, $\Pi$ and $\Pi'$ must have the same distributed complexity, since it is trivial to solve any graph problem on constant-size instances in $O(1)$ rounds.
Thus, if we have a black-box sequential algorithm $\mathcal{A}$ that decides the optimal distributed complexity for an $\LCL$ problem $\Pi'$ given in the node-edge-checkable formalism, then the same algorithm can be applied to an $\LCL$ problem $\Pi$ given in the standard format. Furthermore, if $\mathcal{A}$ also outputs a description of a distributed algorithm solving $\Pi'$, then this distributed algorithm can also be applied to solve $\Pi$, modulo an $O(1)$-round post-processing step.

The number of solvable and unsolvable instances for $\Pi$ and $\Pi'$ are the same for the case of paths, but they might differ by at most an additive constant for the case of cycles.
Suppose we have a black-box sequential algorithm $\mathcal{A}$ that given an $\LCL$ problem $\Pi'$ in the node-edge-checkable formalism, decides 
\[(\# \text{solvable instances}, \# \text{solvable instances}) \in \{(0, \infty),(\Theta(1), \infty),(\infty, \infty),(\infty, \Theta(1)),(\infty, 0)\}.\]
Then obviously the same algorithm can be applied to an $\LCL$ problem $\Pi$ in the standard format for the case of paths.

For solvability on cycles, we can still apply $\mathcal{A}$ to decide the solvability of $\Pi$, but things are a little more complicated as the behavior of $\Pi'$ might be different from $\Pi$ for cycles of length at most $2r+1$. To deal with this issue, instead of applying $\mathcal{A}$ directly on $\Pi'$, we apply $\mathcal{A}$ to a modified $\LCL$ problem $\Pi^\ast$ such that $\Pi^\ast$ is unsolvable on cycles of length at most $2r+1$, and its solvability on longer cycles are the same as that of $\Pi'$.
When the output of $\mathcal{A}$ on $\Pi^\ast$ is $(\Theta(1), \infty)$, $(\infty, \infty)$, or $(\infty, \Theta(1))$, then the same result applies to $\Pi$. If the output is $(0, \infty)$ or $(\infty, 0)$, we just need to further check in polynomial time the number of solvable and unsolvable instances for cycles of length at most $2r+1$ in order to determine the correct solvability of $\Pi$. 
To construct $\Pi^\ast$ from $\Pi'$, we simply let $\Pi^\ast$ be an $\LCL$ that is required to solve $\Pi'$ and another problem $\Pi^\star$ simultaneously, where the $\Pi^\star$ is an arbitrary node-edge-checkable problem that is unsolvable for cycles of length at most $2r+1$, and is solvable for all cycles of length at least $2r+2$.

\paragraph{LCL in standard form.}
Recall that an $\LCL$ problem $\Pi$ may come in many different forms. It may ask for a labeling of nodes, a labeling of edges, a labeling of half-edges, an orientation of the edges, or any combination of these. The canonical way to specify an $\LCL$ with locality radius $r$ is to list all allowed labeled radius-$r$ subgraphs in the set $\mathcal{C}$. An output labeling $\lambda$ for $\Pi$ on the instance $G$ is legal if for each node $v$ in $G$, its radius-$r$ subgraph with the output labeling $\lambda$ belongs to $\mathcal{C}$.

\paragraph{Description length.} From now on, we write $|\Pi|$ to denote the description length of the $\LCL$ problem $\Pi$. For example, if $\Pi$ only asks for an edge orientation, then $|\Pi| = 2^{O(r)}$. If $\Pi$ also asks for an edge labeling from the alphabet $\Sigma_e$ and a node labeling from the alphabet $\Sigma_v$, then $|\Pi| = {(|\Sigma_e|+|\Sigma_v|)}^{O(r)}$. Note that we only consider paths and cycles, and we assume that $\mathcal{C}$ is described using a truth table mapping each labeled radius-$r$ subgraph to yes/no. 

\paragraph{From general labels to half-edge labels.} We first observe that labels of all forms can be transformed into half-edge labels, and so from now on we can assume that $\Pi$ only have half-edge labels. Specifically, if $\Pi$ asks for an edge labeling from the alphabet $\Sigma_e$, a node labeling from the alphabet $\Sigma_v$, and also an edge orientation, then we can simply assume that $\Pi$ asks for a half-edge labeling from the alphabet $\Sigma_e \times \Sigma_v \times \{H,T\}$. That is, each half-edge label is of the form $(a \in \Sigma_e, b \in \Sigma_v, c \in \{H,T\})$.
For each edge $e$, it is required that the $\Sigma_e$-part of the two half-edges of $e$ are the same, and this label represents the edge label of~$e$. 
For each node $v$, it is required that the $\Sigma_v$-part of the two half-edges surrounding $v$ are the same, and this label represents the node label of~$v$. 
For each edge $e$, it is required that the $\{H,T\}$-part of the two half-edges of $e$ are different, and this label represents the edge orientation of~$e$. 
This reduction from a general labeling to a half-edge labeling increases the description length, but only polynomially.

\paragraph{Reducing the locality radius.} We assume that $\Pi$ only asks for a half-edge labeling from the alphabet $\lcllabel$. We will first show a construction of $\Pi'$ in the node-edge-checkable formalism satisfying all the needed requirements.
In what follows, we assume that the communication network $G$ must not be a cycle of at most $2r+1$ nodes. In particular, this ensures that any radius-$r$ subgraph of $G$ is a path, not a cycle.

Each radius-$r$ subpath $P = (u_a, \ldots, u_2, u_1, v, w_1, w_2, \ldots, w_b)$ centered at $v$ with half-edge labels from $\lcllabel$ can be represented as a string $S \in (\lcllabel \cup \{\bot\})^{4r}$, as follows. The string $S$ is of the form $S = S_1 \circ S_2 \circ S_3 \circ S_4$, where
\begin{align*}
S_1 &= \bot^{2(r-a)},\\
S_2 &= L_{u_a,\{u_a, u_{a-1}\}} L_{u_{a-1},\{u_{a}, u_{a-1}\}}
L_{u_{a-1},\{u_{a-1}, u_{a-2}\}}
L_{u_{a-2},\{u_{a-1}, u_{a-2}\}}
\cdots 
L_{v,\{u_{1}, v\}}, \\
S_3 &= L_{v,\{v, w_1\}}
\cdots
L_{w_{a-2},\{w_{a-2}, w_{a-1}\}}
L_{w_{a-1},\{w_{a-2}, w_{a-1}\}}
L_{w_{a-1},\{w_{a-1}, w_{a}\}}
L_{w_a,\{w_{a-1}, u_{a}\}},\\
S_4 &= \bot^{2(r-b)}.
\end{align*}
Here $L_{z,e}$ represents the half-edge label of the edge $e$ at the node $z$. Note that $S_2$ represents the half-edge labels within $(u_a, \ldots, u_2, u_1, v)$, and $S_3$ represents the half-edge labels within $(v, w_1, w_2, \ldots, w_b)$. This string notation is sensitive to the direction of $P$. If $P$ is reversed, then the resulting string $S$ is also reversed.

If we do not care about cycles of length at most $2r+1$, then we can simply assume that the set of allowed configurations $\mathcal{C}$ is specified by a set of strings $S \in (\lcllabel \cup \{\bot\})^{4r}$ in the above form. Note that $S \in \mathcal{C}$ implies that its reverse $S^R$ is also in $\mathcal{C}$.

Now we are ready to describe the new $\LCL$ problem $\Pi'$. In this new $\LCL$ problem, each half-edge label is a string $S \in \mathcal{C}$. We have the following constraints.

\begin{description}
\item[Node constraint:] For each node $v$, let $S$ and $S'$ be the two half-edge labels surrounding $v$, then $S'$ is the reverse of $S$. Furthermore, if $v$ is of degree-1, then the length-$2r$ prefix of $S$ must be $\bot^{2r}$.
\item[Edge constraint:] For each edge $e$, let $S$ and $S'$ be the two half-edge labels of $e$, then the length-$2(r-1)$ suffix of $S$ is the reverse of the length-$2(r-1)$ suffix of $S'$.
\end{description}

Intuitively, a half-edge label $S$ on the $v$-side of the edge $e=\{u,v\}$ is intended to represent the radius-$r$ subgraph $P$ centered at $v$, where the direction of $P$ is chosen as $v \rightarrow u$. The above constraints ensure that the half-edge labels must be consistent.

The transformation from a legal labeling $\lambda$ for $\Pi$ on $G$ to a legal labeling $\lambda'$ for $\Pi'$ on $G$ is straightforward, and it only costs $O(1)$ rounds. The reverse transformation is also straightforward. The description length of $\Pi'$ is $|\lcllabel|^{O(r)}$, which is polynomial in $|\Pi|$.

\begin{figure}
\centering
\includegraphics[page=3]{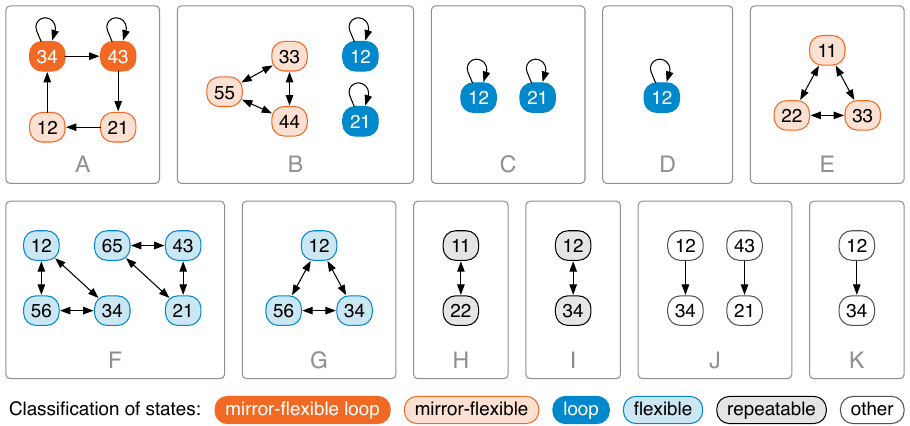}
\caption{Examples of $\LCL$ problems of each type (types A--K in Table~\ref{tab:cycles}) represented as automata, together with a classification of their states using Definitions \ref{def:repeatable-state}--\ref{def:mirror-flexible-loop}. The states are colored only by the most restrictive property. Here is a brief description of each sample problem: \textbf{A:}~orient the edges so that each consistently oriented fragment consists of at least two edges, one with the label pair 12 and at least one with the label pair 34. \textbf{B:}~either find a consistent orientation (encoded with labels 1--2) or find a proper 3-coloring of the edges (encoded with labels 3--5). \textbf{C:}~consistent orientation. \textbf{D:}~orientation in the positive direction. \textbf{E:}~edge 3-coloring. \textbf{F:}~consistent orientation together with an edge 3-coloring. \textbf{G:}~orientation in the positive direction together with an edge 3-coloring. \textbf{H:}~edge 2-coloring. \textbf{I:}~orientation in the positive direction together with an edge 2-coloring. \textbf{J--K:}~problems only solvable on paths of length at most $2$ (assuming appropriate starting and accepting states).}\label{fig:cycles}
\end{figure}

\section{Classification of all \tLCL{} problems on cycles}\label{sec:cycles}

We will now discuss our classification of $\LCL$ problems on cycles. Consider a problem $\Pi$ and its corresponding automation $\M_\Pi$. In what follows, if we have two states $ab$ and $cd$ of $\M_\Pi$, a \emph{walk} from $ab$ to $cd$ (denoted by $ab \leadsto cd$) is a sequence of state transitions starting at state $ab$ and ending at state $cd$. We introduce the following definitions; see Figure~\ref{fig:cycles} for examples.
\begin{definition}[repeatable state]\label{def:repeatable-state}
State $ab \in \CE$ is repeatable if there is a walk $ab \leadsto ab$ in~$\M_\Pi$.
\end{definition}
\begin{definition}[flexible state~\cite{Brandt2017}]\label{def:flexible-state}
State $ab \in \CE$ is flexible with flexibility $K$ if for all $k \ge K$ there is a walk $ab \leadsto ab$ of length exactly $k$ in~$\M_\Pi$.
\end{definition}
\begin{definition}[loop]\label{def:loop}
State $ab \in \CE$ is a loop if there is a state transition $ab \to ab$ in~$\M_\Pi$.
\end{definition}

Observe that each defined property of a state is a proper strengthening of the previous property (i.e.\ each loop is a flexible state and each flexible state is a repeatable state).

For a symmetric problem $\Pi$ we also define:
\begin{definition}[mirror-flexible state]\label{def:mirror-flexible-state}
State $ab \in \CE$ is mirror-flexible with flexibility $K$ if for all $k \ge K$ there are walks $ab \leadsto ab$, $ab \leadsto ba$, $ba \leadsto ab$, and $ba \leadsto ba$ of length exactly $k$ in~$\M_\Pi$.
\end{definition}
\begin{definition}[mirror-flexible loop]\label{def:mirror-flexible-loop}
State $ab \in \CE$ is a mirror-flexible loop with flexibility $K$ if $ab$ is a mirror-flexible state with flexibility $K$ and $ab$ is also a loop.
\end{definition}
Note that if $ab$ is mirror-flexible loop, then so is $ba$, as the problem is symmetric.

\subsection{Flexibility and synchronizing words}

Flexibility is a key concept that we will use in our characterization of $\LCL$ problems. We will now connect it to the automata-theoretic concept of synchronizing words.

First, let us make a simple observation that allows us to study automata by their strongly connected components:
\begin{lemma}
Let $\M'$ be a strongly connected component of automaton $\M_\Pi$, and let $q$ be a state in $\M'$. Then $q$ is flexible in $\M_\Pi$ if and only if $q$ is flexible in $\M'$.
\end{lemma}
\begin{proof}
A walk from $q$ back to $q$ in $\M_\Pi$ cannot leave $\M'$.
\end{proof}

Recall that a word $w$ is called \emph{D3-directing word} \cite{Imreh1999} for NFA $\M$ if there is a state $t$ such that starting with any state $s$ of $\M$ there is a sequence of state transitions that takes $\M$ to state $t$ when it processes~$w$.
We show that this specific notion of a nondeterministic synchronizing word is, in essence, equivalent to the concept of flexibility:
\begin{lemma}
Consider a strongly connected component $\M'$ of some automaton $\M_\Pi$. The following statements are equivalent:
\begin{enumerate}[noitemsep]
    \item There is a flexible state in $\M'$.
    \item All states of $\M'$ are flexible.
    \item There is a D3-directing word for $\M'$.
\end{enumerate}
\end{lemma}
\begin{proof}
(1)$\implies$(2): Assume that state $q$ has flexibility $K$. Let $x$ be another state in $\M'$. As it is in the same connected component, there is some $r$ such that we can walk from $x$ to $q$ and back in $r$ steps. Therefore for any $k \ge K$ we can walk from $x$ back to $x$ in $k+r$ steps by following the route $x \leadsto q \leadsto q \leadsto x$. Hence $x$ is a flexible state with flexibility at most $K+r$.

(2)$\implies$(3): Assume that state $q$ has flexibility $K$, and there is a walk of length at most $r$ from any state $x$ to state $q$. Then we can walk from any state $x$ to $q$ in exactly $r+K$ steps: first in $r' \le r$ steps we can reach $q$ and then in $K+r-r' \ge K$ steps we can walk from $q$ back to itself. Hence $w = \unary^{K+r}$ is a D3-directing word for automaton $\M'$ that takes it from any state to state $q$.

(3)$\implies$(1): Assume that there is some D3-directing word $w = \unary^K$ that can take one from any state of $\M'$ to state $q$ in exactly $K$ steps. Then we can also walk from $q$ to itself in $k$ steps for any $k \ge K$: first take $k-K$ steps arbitrarily inside $\M'$, and then walk back to $q$ in exactly $K$ steps.
\end{proof}

Hence, in what follows, we can freely use any of the above perspectives when reasoning about the distributed complexity of $\LCL$ problems. Mirror-flexibility can be then seen as a mirror-symmetric extension of D3-directing words.

There is also a natural connection between flexibility and \emph{Markov chains}. Automaton $\M_\Pi$ over the unary alphabet can be viewed as the diagram of a Markov chain for unknown probabilities of the transitions. If we assume that every edge will have a non-zero probability, then a strongly connected component of the automaton is an \emph{irreducible} Markov chain, and in such a component the notion of flexibility coincides with the notion of \emph{aperiodicity}.

\subsection{Results}

Our main result is the classification presented in Table~\ref{tab:cycles}; see also Figure~\ref{fig:cycles} for some examples of problems in each class. What was already well-known by prior work \cite{Chang2019,Balliu2019decidable} is that there are only three possible complexities: $O(1)$, $\Theta(\log^* n)$, and $\Theta(n)$. However, our work gives the first concise classification of exactly which problems belong to which complexity class. In Section~\ref{sec:cycles-correct} we show that our classification of locally checkable problems on cycles or paths into types A--K, defined by properties of the automaton, is correct and complete.

The entire classification \emph{can be computed efficiently}. In particular, for all of the defined properties (repeatable states, flexible states, loops, mirror-flexible states and mirror-flexible loops) a polynomial-time algorithm can determine if an automaton contains a state with such a property. The non-trivial cases here are flexibility and mirror-flexibility; we present the proofs in Section~\ref{sec:efficient}.

\subsection{Key building blocks}

\paragraph{The role of mirror-flexibility.}

Consider the following problem that we call \emph{distance-$k$ anchoring}; here the selected edges are called \emph{anchors}:
\begin{definition}\label{def:anchoring}
A distance-$k$ anchoring is a maximal subset of edges that splits the cycle in fragments of length at least $k-1$.
\end{definition}
\noindent This problem can be solved in $O(\log^* n)$ rounds (e.g.\ by applying maximal independent set algorithms in the $k$th power of the line graph of the input graph). Now consider an $\LCL$ problem $\Pi$ that has a flexible state $q$ with flexibility $k$. It is known by prior work \cite{Brandt2017} that we can now solve $\Pi$ on directed cycles in $O(\log^* n)$ rounds, as follows: Solve distance-$k$ anchoring and label the anchor edges with the label pair of state $q$. As state $q$ is flexible, we can walk along the cycle from one anchor to another, and find a way to fill in the fragment between two anchors with a feasible label sequence.

Mirror-flexibility plays a similar role for undirected cycles: the key difference is that the anchor edges cannot be consistently oriented, and hence we need to be able to also fill a gap between state $q = ab$ and its mirror $q' = ba$, in any order. It is easy to see that mirror-flexibility then implies $O(\log^* n)$-round solvability---what is more surprising is that the converse also holds: $O(\log^* n)$-round solvability necessarily implies the existence of a mirror-flexible state.

\paragraph{A new canonical problem for constant-time solvability.}

One of the new conceptual contributions of this work is related to the following problem, which we call \emph{distance-$k$ orientation}:
\begin{definition}\label{def:orientation}
A distance-$k$ orientation is an orientation in which each consistently oriented fragment has a length at least $k$.
\end{definition}
\noindent The problem is trivial to solve in directed cycles in $0$ rounds, but the case of undirected cycles is not equally simple. However, with some thought, one can see that the problem can be solved in $O(1)$ rounds also on undirected cycles \cite{Chang2019}. This shows that there are infinite families of nontrivial $O(1)$-time solvable problems, and hence it seems at first challenging to concisely and efficiently characterize all such problems. However, as we will see in Section~\ref{sec:cycles-correct}, distance-$k$ orientation can be seen as the \emph{canonical $O(1)$-time solvable problem} on undirected cycles. We show that any problem $\Pi$ that is $O(1)$-time solvable on undirected cycles has to be of type A (see Table~\ref{tab:cycles}), and any such problem can be solved in two steps: first find a distance-$k$ orientation for some constant $k$ that only depends on the structure of $\M_\Pi$, and then map the distance-$k$ orientation to a feasible solution of $\Pi$.

We  summarize the key new observations related to undirected cycles as follows:
\begin{framed}
\centering
$\Theta(1)$ rounds $\iff$ mirror-flexible loop $\iff$ solvable with distance-$k$ orientation

$\Theta(\log^* n)$ rounds $\iff$ mirror-flexible state $\iff$ solvable with distance-$k$ anchoring
\end{framed}

\section{Classification of all \tLCL{} problems on paths}\label{sec:paths}

So far we have discussed $\LCL$ problems on cycles; let us now have a look at the case of paths. We have already presented the classification for both cases in Table~\ref{tab:cycles}.
In what follows, we discuss the key new aspects that arise in paths in comparison with the case of cycles.

\paragraph{What is similar: distributed complexity.}

Broadly speaking, \emph{efficient distributed solvability} on paths is not that different from efficient solvability on cycles. Consider an $\LCL$ problem $\Pi$ and the state machine $\M_\Pi$. As a first, preprocessing step, we have to \emph{remove all states that are not reachable from a starting state, and all states from which there is no path to an accepting state}---such states can never appear in any feasible labeling of a path. The removal of irrelevant states can be done in polynomial time, and hence throughout this work we assume that such states have already been eliminated and, to avoid trivialities, the resulting automaton is nonempty.

Now consider, for example, the case of directed paths. If there is a loop $q$ in $\M_\Pi$, we can solve $\Pi$ in constant time. By assumption $q$ can be reached from some starting state $s$ and we can reach some accepting state $t$ from $q$. Hence near the endpoints of a path, we can label according to the walks $s \leadsto q$ and $q \leadsto t$, and fill in everything in between with $q$; the round complexity is simply the maximum of the lengths of the (shortest) walks $s \leadsto q$ and $q \leadsto t$. Similarly, if $q$ is not a loop but a flexible state with flexibility $k$, we can find a distance-$k$ anchoring for the internal part of the path, use $q$ at the anchor points, and fill the gaps just like in the case of a cycle. The case of undirected paths and mirror-flexibility is analogous.

Furthermore, negative results on cycles imply negative results on paths. To see this, consider a hypothetical algorithm $\A$ that solves $\Pi$ efficiently in directed paths. Then we could also apply $\A$ to each local neighborhood of a long directed cycle, and hence $\A$ would also solve $\Pi$ efficiently in directed cycles. If $\Pi$ cannot be solved in $o(n)$ rounds in directed cycles, it cannot be solved in $o(n)$ rounds in directed paths, either. The same holds for the undirected case. Hence the classification of distributed complexities in Table~\ref{tab:cycles} generalizes to paths almost verbatim.

\paragraph{What is new: solvability.}

In directed cycles, global problems (i.e., problems of round complexity $\Theta(n)$, types H and I) came in only one possible flavor: there are infinitely many solvable instances and infinitely many unsolvable instances. A simple example is the problem of finding a proper $2$-coloring: even cycles are solvable and odd cycles are unsolvable. Our classification for cycles implies that it is not possible to have an $\LCL$ problem of complexity $\Theta(n)$ in directed cycles that is always solvable.

This is clearly different in directed paths. As a simple example, $2$-coloring a path is a global problem on directed paths that is always solvable. Figure~\ref{fig:ex-path} shows both examples of $\LCL$s that are solvable in all paths (e.g.\ $2$-coloring), and examples of $\LCL$s that are solvable in infinitely many paths and unsolvable in infinitely many paths (e.g.\ $2$-coloring in which all endpoints must have color $1$). It is also easy to construct problems that are solvable in all but finitely many instances and problems that are solvable only in finitely many instances. However, can we efficiently tell the difference between these cases if we are given a description of an $\LCL$ problem?

This is a question in which the automata-theoretic perspective gives direct answers. In essence, the question is rephrased as follows: for which values of $k$ a nondeterministic finite automaton $\M$ accepts the unary string $\unary^k$; whether $\M$ accepts all such strings is the classical \emph{universality problem} \cite{HOLZER11survey} for unary languages. Prior work directly implies the following (see Section~\ref{app:path-solvability} for the details):
\begin{itemize}
    \item \textbf{\boldmath $0$ vs.\ $\Theta(1)$ unsolvable instances:} Consider the following decision problem: given an automaton $\M$, answer ``yes'' if $\M$ accepts all strings, ``no'' if $\M$ rejects at least one but finitely many strings, and answer ``yes'' or ``no'' otherwise. This problem can be solved in polynomial time, as a consequence of Chrobak's theorem~\cite{AnthonyWidjajaTo09,Chrobak86}.
    \item \textbf{\boldmath $0$ vs.\ $\infty$ unsolvable instances:} Consider the following decision problem: given an automaton $\M$, answer ``yes'' if $\M$ accepts all strings, ``no'' if $\M$ rejects infinitely many strings, and answer ``yes'' or ``no'' otherwise. This is a well-known $\coNP$-complete problem~\cite{Stockmeyer73}.
\end{itemize}

\subsection{Complexity of deciding solvability in paths}\label{app:path-solvability}

We show in Theorem~\ref{thm:chrobak} (see below) that the unary NFA universality problem becomes polynomial-time solvable once we have a promise that $\M$ rejects only finitely many strings. The theorem implies that distinguishing between 0 unsolvable instances and $\Theta(1)$ unsolvable instances is in polynomial time, for both $\LCL$s on paths and on cycles. Although the automaton $\M$ used in the node-edge-checkable formalism has a different acceptance condition than that of the standard NFA, it is straightforward to transform $\M$ into an equivalent NFA with the standard NFA acceptance condition (i.e., there is one starting state $q_0 \in Q$, and a set of accepting states $F$).

\begin{theorem}\label{thm:chrobak}
There is a polynomial-time algorithm $\mathcal{A}$ that achieves the following for any given unary NFA $\M$. If $\M$ does not reject any string, then the output of $\mathcal{A}$ is Yes. If $\M$ rejects at least one but only finitely many strings, then the output of $\mathcal{A}$ is No. If $\M$ rejects infinitely many strings, the output of $\mathcal{A}$ can be either No or Yes.
\end{theorem}
\begin{proof}
This is an immediate consequence of Chrobak's theorem~\cite{AnthonyWidjajaTo09,Chrobak86}, which shows that any unary NFA $\M$ is equivalent to some NFA $\M'$ in the Chrobak normal form, and the number of states in $\M'$ is at most $|Q|^2$. An NFA $\M'$ is in Chrobak normal form if it can be constructed as follows. Start with a directed path $P = (q_0 \rightarrow q_1 \rightarrow \cdots \rightarrow q_m)$ and $k$ directed cycles $C_i = (r_{0,i} \rightarrow r_{1,i} \rightarrow \cdots \rightarrow r_{\ell_i, i} \rightarrow r_{0,i})$, for each $i \in \{1, \ldots, k\}$, where $\ell_i$ is the length of $C_i$. Add a transition from $q_m$ to $r_{0,i}$ for each $i \in \{1, \ldots, k\}$. The starting state is $q_0$. The set of accepting states $F$ can be arbitrary.

The algorithm $\mathcal{A}$ works as follows. It tests whether $\M$ accepts all strings of length at most $|Q|^2$. If so, then the output is Yes; otherwise, the output is No. To see the correctness, we only need to show that whenever $\M$ rejects at least one but only finitely many strings, then the output of $\mathcal{A}$ is No. 
To show this, it suffices to prove that if there is a string $w$ of length higher than $|Q|^2$ that is rejected by $\M$, then there must be infinitely many strings rejected by $\M$. 

Let $L$ be the length of $w$. Now consider some NFA $\M'$ that is in the Chrobak normal form and is equivalent to $\M$. 
We can assume that the number of states in $\M'$ is at most $|Q|^2 < L$.
Define $S$ to be the set of states that is reachable from $q_0$ in $\M'$ via a walk of length exactly $L$.
Since the number of states in $\M'$ is smaller than the length $L$ of $w$, the set $S$ contains exactly one state from each cycle $C_i$. 
Since $w$ is rejected, all states in $S$ are not accepting states.
It is clear that for any non-negative integer $k$, the set of states that is reachable from $q_0$ in $\M'$ via a walk of length exactly $L + k \prod_{1 \leq i \leq k} \ell_i$ is also $S$. Hence $\M'$ (and also $\M$) rejects infinitely many strings.
\end{proof}

Theorem~\ref{thm:nphard} (see below) is a well-known result of $\coNP$-completeness of testing universality (i.e., $\LL(\M) = \Sigma^\ast$) of a unary NFA. 
To see that the same hardness result applies to the analogous question of solvability of $\LCL$s on paths, given any NFA $\M$, we construct a finite state machine $\M^\ast$ representing an $\LCL$ in the node-edge-checkable formalism such that $\M$ is equivalent to $\M^\ast$. For each transition $a \rightarrow b$ in $\M$, add the state $(a,b)$ to $\M^\ast$. Add a transition $(a,b) \rightarrow (c,d)$ in $\M^\ast$ if $b = c$.
Each state $(a,b)$ with $a = q_0$, where $q_0$ is the starting state of $\M$, is designated as a starting state of $\M^\ast$.
Each state $(a,b)$ with $b \in F$, where $F$ is the set of accepting states of $\M$, is designated as an accepting state of $\M^\ast$. Now the new finite state machine $\M^\ast$ represents an $\LCL$ on paths.
Note that this reduction only works for $\LCL$ on paths---the same solvability problem on cycles can be solved in polynomial time.

\begin{theorem}[Stockmeyer and Meyer \cite{Stockmeyer73}]\label{thm:nphard}
Given a unary NFA $\M$, the following problem is $\NP$-hard. If $\M$ reject zero string, then the output is required to be No. If $\M$ rejects at least one but only finitely many strings, then the output can be either No or Yes. If $\M$ rejects infinitely many strings, the output is required to be Yes.
\end{theorem}

\section{Efficient computation of the classification of \tLCL{} problems}\label{sec:efficient}

In view of Table~\ref{tab:cycles}, the task to classify for an $\LCL$ problem $\Pi$ to which class it belongs to can be reduced to testing certain graph properties of $\M_\PP$. In this section, we show that checking whether a state $q$ is flexible or mirror-flexible can be done in polynomial time, and so deciding the optimal distributed complexity of an $\LCL$ problem $\PP$ is also in polynomial time.

\begin{definition}\label{def:Lq}
Let $Q$ be the set of states of $\M$. For each $q \in Q$ we define:
\begin{itemize}[noitemsep]
    \item $L_q$ is the set of values $\ell$ such that there is a walk $q \leadsto q$ of length $\ell$ in $\M$.
    \item $L'_q = \{ \ell \in L_q : \ell \le 2|Q|-1 \}$ is the restriction of $L_q$ to walks of length at most $2|Q|-1$.
\end{itemize}
\end{definition}

\begin{lemma} \label{lem:gcd_of_returning_walks}
For any automaton $\M$ and for any state $u$, we have $\gcd(L_u) = \gcd(L_u')$.
\end{lemma}
\begin{proof}
We show that for each $\ell \in L_u \setminus L'_u$, we can find $\ell_1,\ell_2,\ell_3 \in L_u$ such that $\ell_1,\ell_2,\ell_3 < \ell$ and $\ell = x \ell_1 + y \ell_2 + z \ell_3$ for some integers $x,y,z$. By applying this argument recursively to each $\ell_i$, we can eventually write any $\ell \in L_u$ as a linear combination of sufficiently small numbers $\ell' \in L'_u$. Hence if all values in $L'_q$ are multiples of some $d$, all values in $L_q$ have to be also multiples of $d$.

Therefore it suffices to show that for each walk $w$ of the form $u \leadsto u$ of length $\ell > 2|Q| - 1$, it is possible to find shorter returning walks $w_1, w_2, w_3$ of the form $u \leadsto u$ of lengths $\ell_1,\ell_2,\ell_3 < \ell$ such that $\ell = x \ell_1 + y \ell_2 + z \ell_3$ for some integers $x,y,z$.

We write $w = (v_1, v_2, \ldots, v_{\ell}, v_{\ell+1})$, where $v_1 = v_{\ell+1} = u$. Since this vector has $\ell+1 \geq 2n+1$ elements, by the pigeonhole principle, there exists a state $v$ that appears at least three times. Therefore, $w$ can be decomposed into four walks: $p_1 = (v_1, \ldots, v_i)$, $p_2 = (v_i, \ldots, v_j)$, $p_3=(v_j, \ldots, v_k)$, and $p_4 = (v_k, \ldots, v_{\ell+1})$, where $v_i = v_j = v_k$ and $1 \leq i < j < k \leq \ell+1$. We write $L_i$ to denote the length of $p_i$. 

Now define $w_1 = p_1 \circ p_4$, $w_2 = p_1 \circ p_2 \circ p_4$, and $w_3 = p_1 \circ p_3 \circ p_4$; the lengths of these paths are $\ell_1 = L_1 + L_4$, $\ell_2 = L_1 + L_2 + L_4$, and $\ell_3 = L_1 + L_3 + L_4$. Now the length $\ell$ of $w$ can be expressed as $\ell = -w_1 + w_2 + w_3$. Since $L_2 = j-i \geq 1$ and $L_3 = k-j \geq 1$, the three lengths $\ell_1, \ell_2, \ell_3$ are all smaller than $\ell$, as required.
\end{proof}

\begin{lemma}\label{lem:gcd_flexible}
A state $q$ is flexible if and only if $\gcd(L_q) = 1$.
\end{lemma}
\begin{proof}
If $\gcd(L_q) = x > 1$, then $kx+1 \notin L_q$ and hence there is no walk $q \leadsto q$ of length $kx+1$ for any $k$, and $q$ cannot be flexible.

For the other direction, given a set of positive integers $S$ with $\gcd(S) = 1$, the \emph{Frobenius number} $g(S)$ of the set $S$ is the largest number $x$ such that $x$ cannot be expressed as a linear combination of $S$, where each coefficient is a non-negative integer. It is known that $g(S) < \max(S)^2$~\cite{shallit2008frobenius}.

By Lemma~\ref{lem:gcd_of_returning_walks}, $\gcd(L_q) = \gcd(L_q')$ and $\max(L_q') \le 2|Q|-1$. Hence $\gcd(L_q) = 1$ implies that for all $k \geq (2|Q|-1)^2$, it is possible to find a length-$k$ walk $q \leadsto q$ by combining some returning walks of length at most $2|Q|-1$, and so $q$ is flexible.
\end{proof}

We remark that the problem of calculating the Frobenius number when the input numbers can be encoded in binary is $\NP$-hard~\cite{ramirez1996complexity}. However, the flexibility of a given automaton can be nevertheless found efficiently.

\begin{lemma} \label{lem:finding_flexible_node_is_in_p}
Testing whether a state $q \in Q$ is flexible and finding its flexibility number is solvable in polynomial time.
\end{lemma}
\begin{proof}
By Lemma~\ref{lem:gcd_flexible}, it is sufficient to test if $\gcd(L_u) = 1$, and by Lemma~\ref{lem:gcd_of_returning_walks}, it suffices to find the set $L_u'$ and compute its $\gcd(L_u')$, which can be done in polynomial time.
\end{proof}

\begin{lemma} \label{finding_mr_flexible_node_is_in_p}
Testing whether a state $q \in Q$ is mirror-flexible and finding its mirror-flexibility number is solvable in polynomial time.
\end{lemma}
\begin{proof}
Follows from Lemma~\ref{lem:finding_flexible_node_is_in_p}: $q \in Q$ is mirror-flexible if and only if $q$ is flexible and is reachable to its mirror $q'$ and $q$ can be reached back from $q'$. Reachability between two states can be tested in polynomial time.
\end{proof}

\begin{theorem} \label{deciding_round_complexity_in_p}
Given an $\LCL$ problem $\Pi$, classifying its type can be computed in polynomial time.
\end{theorem}
\begin{proof}
The non-trivial cases are captured in Lemmas \ref{lem:finding_flexible_node_is_in_p} and \ref{finding_mr_flexible_node_is_in_p}.
\end{proof}

An immediate corollary of Lemma~\ref{lem:finding_flexible_node_is_in_p} is that the existence of a D3-directing word in an NFA over the unary alphabet can be decided in polynomial time, since a unary NFA $\M$ has  a D3-directing word if and only if there exists a strongly connected component $S$ such that all states in $\M$ are reachable to $S$ and the subgraph induced by $S$ admits a flexible state. For comparison, when the alphabet size is at least two, the same problem is known to be  $\PSPACE$-hard~\cite{Martyugin2014}.

\begin{corollary}
Given an NFA $\M$ over the unary alphabet, the existence of a D3-directing word    can be decided in polynomial time.
\end{corollary}

\section{Correctness of the classification of \tLCL{} problems}\label{sec:cycles-correct}

In this section, we show that the classification of $\LCL$ problems in Table~\ref{tab:cycles} is correct and complete. We first prove the round complexity of each type and then the solvability. The connection between the proofs and the results they establish is depicted in Table~\ref{tab:pointers}.

\subsection{Round complexity lower bounds}

\begin{table}
\newcommand{\thref}[1]{\ref{#1}}
\newcommand{\triv}{triv.}
\newcommand{\na}{---}
\newcommand{\colsp}{\hspace{3ex}}
\newcommand{\mythm}{Theorem:}
\centering
\begin{tabular}{@{}l@{\colsp\colsp\colsp}l@{\colsp}l@{\colsp}l@{\colsp}l@{\colsp}l@{\colsp}l@{\colsp}l@{\colsp}l@{}}
\toprule
Type:
&& \textbf{A} & \textbf{B} & \textbf{C}/\textbf{D} & \textbf{E} & \textbf{F}/\textbf{G} & \textbf{H}/\textbf{I} & \textbf{J}/\textbf{K} \\
\midrule
\multicolumn{8}{@{}l}{\em Number of instances:}\\
\addlinespace[2pt]
$\cdot$ {solvable cycles}     &\mythm&\thref{thm:reps_inf_sol}  &\thref{thm:reps_inf_sol}  &\thref{thm:reps_inf_sol}  &\thref{thm:reps_inf_sol}  &\thref{thm:reps_inf_sol}  &\thref{thm:reps_inf_sol}                &\thref{thm:cycl_no_rep_0_sol} \\
$\cdot$ {solvable paths}      &\mythm&\thref{thm:reps_inf_sol}  &\thref{thm:reps_inf_sol}  &\thref{thm:reps_inf_sol}  &\thref{thm:reps_inf_sol}  &\thref{thm:reps_inf_sol}  &\thref{thm:reps_inf_sol}                &\thref{thm:no_rep_impl_c_sol} \\
$\cdot$ {unsolvable cycles}   &\mythm&\thref{thm:loops_0_unsol} &\thref{thm:loops_0_unsol} &\thref{thm:loops_0_unsol} &\thref{thm:flexs_c_unsol} &\thref{thm:flexs_c_unsol} &\thref{thm:cycles_only_reps_inf_unsol}  &\thref{thm:no_rep_impl_c_sol} \\
$\cdot$ {unsolvable paths}    &\mythm&\thref{thm:flexs_c_unsol} &\thref{thm:flexs_c_unsol} &\thref{thm:flexs_c_unsol} &\thref{thm:flexs_c_unsol} &\thref{thm:flexs_c_unsol} &\na                                     &\thref{thm:no_rep_impl_c_sol} \\
\midrule
\multicolumn{8}{@{}l}{\em Round complexity for directed graphs:}\\
\addlinespace[2pt]
$\cdot$ {lower bound}        &\mythm&\triv  &\triv   &\triv   & \thref{thm:dir_no_loop_state_impl_log*}&\thref{thm:dir_no_loop_state_impl_log*} &\thref{thm:no_flexible_impl_Omega(n)}  &\triv     \\
$\cdot$ {upper bound}        &\mythm&\thref{thm:dir_loop_impl_O(1)}  &\thref{thm:dir_loop_impl_O(1)}   &\thref{thm:dir_loop_impl_O(1)}      &\thref{thm:dir_flexible_impl_O(log*n)}&\thref{thm:dir_flexible_impl_O(log*n)} &\triv  &\thref{thm:no_rep_impl_O(1)} \\
\midrule
\multicolumn{8}{@{}l}{\em Round complexity for undirected graphs:}\\
\addlinespace[2pt]
$\cdot$ {lower bound}       &\mythm&\triv  & \thref{thm:undir_no_mfloop_impl_Omega(log*)}&\thref{thm:undir_no_mflexible_impl_Omega(n)}&\thref{thm:undir_no_loop_state_impl_log*} &\thref{thm:undir_no_mflexible_impl_Omega(n)}&\thref{thm:no_flexible_impl_Omega(n)}&\triv     \\
$\cdot$ {upper bound}       &\mythm&\thref{thm:undir_mfloop_impl_O(1)}  &\thref{thm:undir_mirror_flexible_impl_O(log*n)}&\triv&\thref{thm:undir_mirror_flexible_impl_O(log*n)}&\triv&\triv&\thref{thm:no_rep_impl_O(1)} \\
\bottomrule
\end{tabular}
\caption{Connection from problem types to proofs establishing their correctness---cf.\ Table~\ref{tab:cycles}.}\label{tab:pointers}
\end{table}

In all proofs in this section, we need a technical assumption that $\M_\PP$ contains a repeatable state. This ensures that for every number $N$, we can find an $n$-node solvable instance $G$ for some $n \geq N$. This assumption is necessary: If $\M_\PP$ does not contain a repeatable state, then we can find a number $N$ such that for all $n \geq N$ the problem $\PP$ has no solution on a cycle or path of $n$ nodes, and so the round complexity of $\PP$ is trivially $O(1)$ in all solvable instances.

\begin{theorem} \label{thm:dir_no_loop_state_impl_log*}
Let $\Pi$ be an $\LCL$ problem on directed cycles or paths. Suppose that the automaton $\M_\Pi$ contains a repeatable state, but it does not contain a loop. Then the round complexity $\Pi$ is $\Omega(\log^*n)$.
\end{theorem}

\begin{proof}
We show how to turn any legal labeling $\lambda$ of $\Pi$ into an edge $3$-coloring in a constant number of rounds. As $3$-coloring of edges requires $\Omega(\log^\ast n)$ rounds~\cite{Linial1992}, so does $\Pi$.

Let $Q$ be the set of states of $\M_\Pi$, and consider a valid solution $\lambda$ of $\Pi$. Such a labeling can be easily turned into an edge $|Q|$-coloring $f$: an edge that was labeled with the pair $(a,b)$ in $\lambda$ will be colored with the color $(a,b)$ in $f$. As there are no loops in $\M_\Pi$, adjacent edges must have different label pairs and hence different colors. Finally, we can reduce the number of colors from $|Q|$ to $3$ in a constant number of rounds (w.r.t.\ to $n$) with the trivial algorithm that eliminates colors one at a time.
\end{proof}

\begin{figure}
\centering
\includegraphics[page=4]{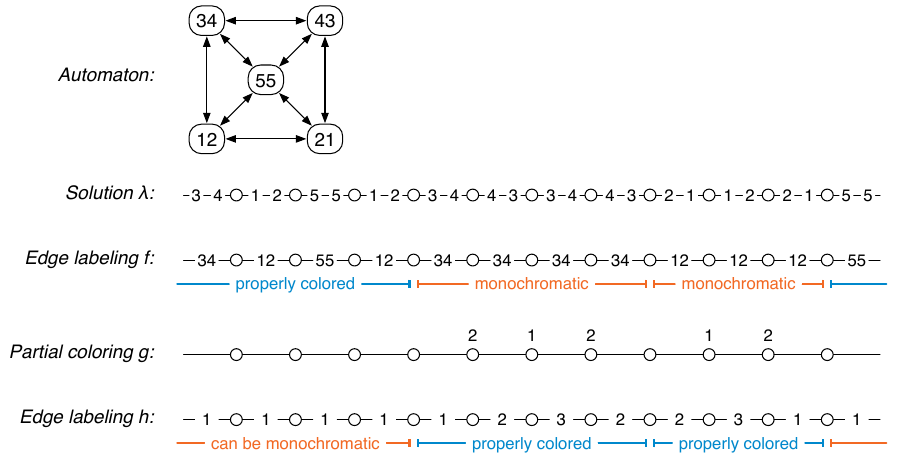}
\caption{An illustration of the proof of Theorem~\ref{thm:undir_no_loop_state_impl_log*}. Pairs $(f(e),h(e))$ form a proper edge coloring.}\label{fig:symmetrybreaking}
\end{figure}

\begin{theorem} \label{thm:undir_no_loop_state_impl_log*}
Let $\Pi$ be an $\LCL$ problem on undirected cycles or paths. Suppose that the automaton $\M_\Pi$ contains a repeatable state, but it does not contain a loop. Then the round complexity $\Pi$ is $\Omega(\log^*n)$.
\end{theorem}
\begin{proof}
We use an idea similar to Theorem~\ref{thm:dir_no_loop_state_impl_log*}, with one extra ingredient. Assume that $\lambda$ is a feasible solution of $\Pi$. First construct a labeling of the edges with (at most) $|Q|$ colors as follows: an edge that was labeled with the pair $(a,b)$ in $\lambda$ will be colored with the color $\{a,b\}$ in $f$ (note that the colors are now unordered pairs). 

Now such a labeling $f$ is not necessarily a proper coloring. There may be an arbitrarily long sequence of edges that have the same label $\{a,b\}$, for some $a < b$; such a path is called \emph{monochromatic}. However, this would arise only if $\lambda$ contains a sequence of the form $ab,\allowbreak ba,\allowbreak ab,\allowbreak ba, \dotsc$. Within such a path, we can find a partial labeling of the nodes $g$ as follows: nodes that have both ports labeled with $a$ are colored with $1$, and nodes that have both ports labeled with $b$ are colored with $2$; all other nodes are left uncolored. See Figure~\ref{fig:symmetrybreaking} for an illustration.

Now we have two ingredients: a not-necessarily-proper edge coloring $f$ with $|Q|$ colors, and a partial node coloring $g$ with $2$ colors. These complement each other: all internal nodes in monochromatic paths of $f$ are properly $2$-colored in $g$. Hence we can use $g$ to find a proper edge $3$-coloring $h$ of each monochromatic path, e.g.\ as follows: Nodes of color $1$ are active and send proposals to adjacent nodes of color $2$ (proposals are sent in the order of unique identifiers), nodes of color $2$ accept the first proposal that they get (breaking ties with unique identifiers), and this way we can find a maximal matching within each monochromatic path. Each such matching forms one color class in $h$; we delete the edges that are colored and repeat. After three such iterations all internal edges of monochromatic paths are properly colored in $h$; then $h$ is easy to extend so that also the edges near the endpoints of monochromatic paths have colors different from their monochromatic neighbors (monochromatic paths of length two are also easy to $3$-color). Now the pairs $(f(e), h(e))$ form a proper edge coloring with $3|Q|$ colors, and we can finally reduce the number of colors down to $3$.
\end{proof}

In both of the following lemmas to be applicable also to the case of a path, we always assume that the ``witness'' of any specific behavior happens somewhere in the middle of a cycle or a path and not next to the endpoints.

\begin{lemma} \label{lem:non_flexible_impl_O(n)}
Let $\Pi$ be an $\LCL$ problem that is solvable in cycles or paths of length $n$ for infinitely many values of $n$. Assume that $\A$ solves $\Pi$ in for all solvable instances, and assume that for arbitrarily large values of $n$, we can find a cycle or a path of length $n$ such that there are two edges $e_1$ and $e_2$ with the following properties:
\begin{itemize}[noitemsep]
    \item The distance between $e_1$ and $e_2$, and the distance between each $e_i$ and the nearest degree-$1$ node (if any) is more than $n/10$.
    \item Algorithm $\A$ labels both $e_1$ and $e_2$ with the same state $q$ that is not flexible.
\end{itemize}
Then the round complexity of $\A$ has to be $\Omega(n)$.
\end{lemma}

\begin{proof}
We give the proof for the case of a path; the case of a cycle is similar. To reach a contradiction, assume the complexity of $\A$ is sublinear. Pick a sufficiently large $n$ such that the algorithm runs in $r \ll n/20$ rounds and paths of length $n$ are solvable. Decompose the path $G$ in fragments
\[G = (P_0, N_1, P_1, N_2, x, P_2),\]
where $N_i$ is the radius-$r$ neighborhood of $e_i$, each $P_i$ is a path of nodes, and $x$ is one node. Now we can move one node to construct another path
\[G' = (P_0, N_1, P_1, x, N_2, P_2).\]
Path $G'$ has the same length as $G$, and hence $G'$ is also a solvable instance and $\A$ has to be able to find a feasible solution. As the radius-$r$ neighborhoods of $e_1$ and $e_2$ are the same in $G$ and $G'$, algorithm $\A$ will label them with $q$ in both $G$ and $G'$. But as $q$ is not flexible, we can this way eventually construct an instance in which the distance between the two edges with label $q$ is $k$ such that $\M_\Pi$ does not have a walk of length $k$ from $q$ back to itself, and hence $\A$ cannot produce a valid solution.
\end{proof}

\begin{lemma} \label{lem:q_without_mirror_impl_O(n)}
Let $\Pi$ be a symmetric $\LCL$ problem that is solvable in undirected cycles or paths of length $n$ for infinitely many values of $n$. Assume that $\A$ solves $\Pi$ in for all solvable instances, and assume that for arbitrarily large values of $n$, we can find a cycle or a path of length $n$ such that there is an edge $e_1$ with the following properties:
\begin{itemize}[noitemsep]
    \item The distance between $e_1$ and the nearest degree-$1$ node (if any) is more than $n/10$.
    \item Algorithm $\A$ labels both $e_1$ with a state $q_1$ that is not mirror-flexible.
\end{itemize}
Then the round complexity of $\A$ has to be $\Omega(n)$.
\end{lemma}

\begin{proof}
We give the proof for the case of a path; the case of a cycle is similar. To reach a contradiction, assume the complexity of $\A$ is sublinear. Pick a sufficiently large $n$ such that the algorithm runs in $r \ll n/20$ rounds and paths of length $n$ are solvable. For the purposes of this proof, orient the path so that the distance between $e_1$ and the end of the path is at least $n/2$. Let $e_2$ be an edge between $e_1$ and the end of the path such that the distance between $e_1$ and $e_2$, and the distance between $e_2$ and the endpoint is at least $n/10$. Decompose the path $G$ in fragments
\[G = (P_0, N_1, P_1, N_2, P_2),\]
where $N_i$ is the radius-$r$ neighborhood of $e_i$, and each $P_i$ is a path of nodes. Let $\bar N_1$ be the mirror image of path $X$, i.e., the same nodes in the opposite direction; then $\A$ will label the midpoint of $\bar N_1$ with $q'_1$, the mirrored version of state $q_1$. Construct the following paths:
\begin{align*}
G_1 &= (P_0, N_2, P_1, N_1, P_2), \\
G_2 &= (P_0, \bar N_1, P_1, N_2, P_2), \\
G_3 &= (P_0, N_2, P_1, \bar N_1, P_2).
\end{align*}
Now all such paths have length $n$, and hence they are also solvable and $\A$ is expected to produce a feasible solution. Such a solution in $G$ gives a walk $q_1 \leadsto q_2$ in $\M_\Pi$, $G_1$ gives a walk $q_2 \leadsto q_1$, $G_2$ gives a walk $q'_1 \leadsto q_2$, and $G_3$ gives a walk $q_2 \leadsto q'_1$. Putting these together, we can construct walks $q_1 \leadsto q_1$, $q_1 \leadsto q'_1$, $q'_1 \leadsto q_1$, and $q'_1 \leadsto q'_1$.

Finally, we can move nodes one by one from $P_2$ to $P_1$ in each of $G, G_1, G_2, G_3$ to construct such walks of any sufficiently large length. It follows that $q_1$ is mirror-flexible, which is a contradiction.
\end{proof}

\begin{theorem} \label{thm:no_flexible_impl_Omega(n)}
Let $\Pi$ be an $\LCL$ problem. Suppose that the automaton $\M_\Pi$ contains a repeatable state, but it does not contain a flexible state. Then the round complexity $\Pi$ is $\Omega(n)$.
\end{theorem}

\begin{proof}
We can apply Lemma~\ref{lem:non_flexible_impl_O(n)}: the algorithm can only use non-flexible states, and it has to use some non-flexible state repeatedly.
\end{proof}

\begin{theorem} \label{thm:undir_no_mfloop_impl_Omega(log*)}
Let $\Pi$ be a symmetric $\LCL$ problem on undirected cycles or paths. Suppose that $\M_\Pi$ contains a repeatable state, but it does not contain a mirror-flexible loop. Then the round complexity of $\Pi$ is $\Omega(\log^*n)$.
\end{theorem}

\begin{proof}
Consider an algorithm $\A$ that solves $\Pi$, and look at the behavior of $\A$ in sufficiently large instances, far away from the endpoints of the paths (if any). There are two cases:
\begin{enumerate}
    \item Algorithm $\A$ sometimes outputs a loop state (which by assumption cannot be mirror-flexible). Then by Lemma~\ref{lem:q_without_mirror_impl_O(n)} we obtain a lower bound of $\Omega(n)$.
    \item Otherwise $\A$ essentially solves the restriction of $\Pi$ where loop states are not allowed (except near the endpoints of the path), and we can use Theorem~\ref{thm:undir_no_loop_state_impl_log*} to obtain a lower bound of $\Omega(\log^*n)$. \qedhere
\end{enumerate}
\end{proof}

\begin{theorem} \label{thm:undir_no_mflexible_impl_Omega(n)}
Let $\Pi$ be a symmetric $\LCL$ problem on undirected cycles or paths. Suppose that $\M_\PP$ contains a repeatable state, but it does not have a mirror-flexible state. Then the round complexity of $\Pi$ is $\Omega(n)$.
\end{theorem}

\begin{proof}
Again consider an algorithm $\A$ that solves $\Pi$, and look at the behavior of $\A$ in sufficiently large instances, far away from the endpoints of the paths (if any). There are two cases:
\begin{enumerate}
    \item Algorithm $\A$ sometimes outputs a flexible state (which by assumption cannot be mirror-flexible). Then by Lemma~\ref{lem:q_without_mirror_impl_O(n)} we obtain a lower bound of $\Omega(n)$.
    \item Otherwise $\A$ essentially solves the restriction of $\Pi$ where flexible states are not allowed (except near the endpoints of the path), therefore it is using some non-flexible state repeatedly far from endpoints, and Lemma~\ref{lem:non_flexible_impl_O(n)} applies. \qedhere
\end{enumerate}
\end{proof}

\subsection{Round complexity upper bounds}

Let us first consider the trivial case of automata without repeating states.

\begin{theorem} \label{thm:no_rep_impl_O(1)}
Let $\Pi$ be an $\LCL$ problem. Suppose $\M_\Pi$ does not have repeatable state. Then $\Pi$ can be solved in constant time in solvable instances.
\end{theorem}
\begin{proof} 
Let $Q$ be a set of states of $\M_\Pi$. As $\M_\Pi$ does not have a repeatable state, it is not solvable in any cycle, and it is only solvable in some paths of length at most $|Q|$. Hence $\Pi$ can be solved in constant time by brute force (and also in constant time all nodes can detect if the given instance is solvable.
\end{proof}

In the rest of this section, we design efficient algorithms for solving problems with flexible or mirror-flexible states. We present the algorithms first for the case of a cycle. The case of a path is then easy to solve: we can first label the path as if it was a cycle, remove the labels near the endpoints (up to distance $k$, where $k$ is bounded by the (mirror-)flexibility of a chosen (mirror-)flexible state plus the number of states in $\M_\Pi$), and fill constant-length path fragments near the endpoints by brute force. We refer to this process as \emph{fixing the ends}.

\begin{theorem} \label{thm:dir_loop_impl_O(1)}
Let $\Pi$ be an $\LCL$ problem on directed cycles or paths. Suppose $\M_\Pi$ has a loop. Then the round complexity $\Pi$ is $O(1)$.
\end{theorem}
\begin{proof}
All edges can be labeled by a loop state. In a path we will then fix the ends.
\end{proof}

\begin{theorem} \label{thm:undir_mfloop_impl_O(1)}
Let $\Pi$ be an $\LCL$ problem. Suppose $\M_\Pi$ has a mirror-flexible loop. Then the round complexity $\Pi$ is $O(1)$.
\end{theorem}
\begin{proof}
Let $q$ be a mirror-flexible loop state of mirror-flexibility $k$. Let $K \ge k+2$ be an even constant. The first step is to construct a distance-$K$ orientation (Definition~\ref{def:orientation}); this can be done in $O(1)$ rounds.

We say that an edge $e$ is a \emph{boundary edge} if there is another edge $e'$ with a different orientation within distance less than $K/2$ from $e$; otherwise $e$ is an \emph{internal edge}. Note that each consistently oriented fragment contains at least one internal edge.

The internal edges are labeled as follows: each edge with orientation ``$\rightarrow$'' is assigned label $q$, and each edge with orientation ``$\leftarrow$'' is assigned label $q'$, i.e., the mirror of $q$.

We are left with gaps of length $K - 2 \ge k$ between the labeled edges. As $q$ is mirror-flexible, we can find paths $q \leadsto q'$ and $q' \leadsto q$ of length $K - 2$ to fill in such gaps. Finally, in a path we will fix the ends.
\end{proof}

\begin{theorem} \label{thm:dir_flexible_impl_O(log*n)}
Let $\Pi$ be an $\LCL$ problem on directed cycles or paths. Suppose $\M_\Pi$ has a flexible state. Round complexity of such $\Pi$ is $O(\log^*n)$.
\end{theorem}

\begin{proof}
Let $q$ be a flexible state of flexibility $k$. This time we first construct a distance-$k$ anchoring (Definition~\ref{def:anchoring}); this can be done in $O(\log^* n)$ rounds. Let the set of anchors be $I$. If an edge is in $I$, we label its ports by $q$. We are left with the gaps, which can be of size between $k-1$ and $2k$ (anchoring is maximal). As $q$ is flexible, for each gap of size $g \geq k - 1$ we can find a returning walk of length exactly $g+1 \geq k$ and fill it by the states along such walk. Finally, in a path we will fix the ends.
\end{proof}

\begin{theorem} \label{thm:undir_mirror_flexible_impl_O(log*n)}
Let $\Pi$ be an $\LCL$ problem on undirected cycles or paths. Suppose $\M_\Pi$ has a mirror-flexible state. Round complexity of such $\Pi$ is $O(\log^*n)$.
\end{theorem}

\begin{proof}
The proof is very similar to a previous proof, only with some minor changes as now we are in the undirected setting.

Let $q$ be a mirror-flexible state of flexibility $k$. First, we construct a distance-$k$ anchoring (Definition~\ref{def:anchoring}); this can be done in $O(\log^* n)$ rounds. Let the set of anchors be $I$. If an edge is in $I$, we label its ports by either $q$ or its mirror $q'$ arbitrarily (breaking symmetry with unique identifiers). We are left with the gaps, which can be of size between $k-1$ and $2k$ (anchoring is maximal). As $q$ is mirror-flexible, for each gap of size $g \geq k-1$ we can find a returning walk of length exactly $g+1 \geq k$ and fill the gap no matter the combinations of anchors ($q \leadsto q'$, $q \leadsto q$, $q \leadsto q'$ or $q' \leadsto q'$). Finally, in a path we will fix the ends.
\end{proof}

\subsection{Solvability}\label{app:solvability}

In this part, we consider the \emph{solvability} of an $\LCL$ problem. That is, for a given graph class $\GG$ (the set of all cycles of every length or the set of paths of every length), how many graphs $G \in \GG$ are solvable instances (instances that admit a legal labeling) with respect to the given $\LCL$ problem $\PP$.

\begin{theorem} \label{thm:reps_inf_sol}
Let $\Pi$ be an $\LCL$ problem. If $\M_\Pi$ has a repeatable state, then the number of solvable instances is $\infty$.
\end{theorem}
\begin{proof}
Let $q$ be a repeatable state, i.e., there is a walk $q \leadsto q$ of some length $\ell$. Now for every $k \in \N$, cycles of length $k\ell$ are solvable, as we can generate cycles of the form $q \leadsto q \leadsto \dotsb$.

In paths, by assumption $q$ is reachable from some starting state $s$ and we can reach some accepting state $t$ from $q$; let $h$ be the length of a walk $s \leadsto q \leadsto t$. Now for every $k \in \N$, paths of length $h + k\ell$ are solvable, as we can generate paths of the form $s \leadsto q \leadsto q \leadsto \dotsb \leadsto q \leadsto t$.
\end{proof}

\begin{theorem} \label{thm:flexs_c_unsol}
Let $\Pi$ be an $\LCL$ problem. If $\M_\Pi$ has a flexible state, number of unsolvable instances is at most $C$, where $C$ is a constant.
\end{theorem}
\begin{proof}
Let $q$ be a flexible state with flexibility $k$. All cycles of length $n \ge k$ are now trivially solvable, as we have a walk $q \leadsto q$ of length $n$.

In paths, by assumption $q$ is reachable from some starting state $s$ and we can reach some accepting state $t$ from $q$; let $h$ be the length of a walk $s \leadsto q \leadsto t$. Now all paths of length $n \ge h+k$ are solvable, as we have a walk $s \leadsto q \leadsto q \leadsto t$ of length $n$.
\end{proof}

\begin{theorem} \label{thm:loops_0_unsol}
Let $\Pi$ be an $\LCL$ problem on cycles. If $\M_\Pi$ has a loop the number of unsolvable instances is zero.
\end{theorem}

\begin{proof}
As $\M_\Pi$ has a loop, returning walks of all lengths exists and all cycles can be labeled.
\end{proof}

\begin{theorem} \label{thm:cycl_no_rep_0_sol}
Let $\Pi$ be an $\LCL$ problem on cycles. If $\M_\Pi$ has does not have a repeatable state the number of solvable instances is zero.
\end{theorem}
\begin{proof}
Any legal labeling on cycles has to contain a repeatable state.
\end{proof}

\begin{theorem} \label{thm:cycles_only_reps_inf_unsol}
Let $\Pi$ be an $\LCL$ problem. Assume $\M_\Pi$ does not have any flexible state. Then there are infinitely many unsolvable instances on cycles.
\end{theorem}
\begin{proof}
Let $Q$ be the set of states of $\M_\Pi$. Since no state is flexible in $\M_\PP$, by Lemma~\ref{lem:gcd_flexible} we have $\gcd(L_q) > 1$ for all $q \in Q$. Pick \[b = \prod_{q\in Q} \gcd(L_q).\]
Now $kb + 1 \notin L_q$ for any $q \in Q$ and any natural number $k$. Therefore it is not possible to use any state $q$ in a cycle of length $kb+1$, as a feasible solution in such a cycle would form a walk $q \leadsto q$ of length $kb+1$. Hence there are infinitely many unsolvable instances.
\end{proof}

\begin{theorem} \label{thm:no_rep_impl_c_sol}
Let $\Pi$ be an $\LCL$ problem. Suppose $\M_\Pi$ does not have repeatable state. Then there are at most constantly many solvable instances.
\end{theorem}

\begin{proof} 
Let $Q$ be a set of states of $\M_\Pi$. As $\M_\Pi$ does not have a repeatable state, all walks that can form legal labeling have to have to have length at most $|Q|$. So all paths of lengths $n >|Q|$ are unsolvable instances.
\end{proof}

\section{Extension to rooted trees}\label{sec:trees}

In the previous sections, we have presented a complete classification of all $\LCL$ problems on paths and cycles. Now we will demonstrate how to leverage these results (in particular, the classification of $\LCL$s on directed paths) to also classify a family of $\LCL$ problems on \emph{rooted trees}.

There is one obstacle to keep in mind: if the $\LCL$ problem can refer to e.g.\ the degrees of the nodes, then we can use the structure of the tree to encode input labels; and then the setting is at least as general as $\LCL$ problems on labeled paths, which implies that questions on locality are at least $\PSPACE$-hard \cite{Balliu2019decidable}. Hence to have efficient classification algorithms for rooted trees we need to choose a restricted family of $\LCL$ problems. We will here use \emph{edge-checkable} problems as an example---as we will see, it is a broad enough family of problems to capture many interesting problems but restricted enough that we can still classify all such problems.

\paragraph{Edge-checkable \tLCL{} problems.}

An \emph{edge-checkable} $\LCL$ problem $\Pi$ on rooted trees consists of a finite set $\Gamma$ of output labels and a set of constraints  $\CE \subseteq \Gamma \times \Gamma$ specifying the set of allowed ordered pairs of labels $(\lambda(u), \lambda(v))$ on the two endpoints $u$ and $v$ of each  edge, where $u$ is the parent of $v$. For example, the vertex $k$-coloring problem is a (symmetric) edge-checkable $\LCL$ problem with $\Gamma=\{1,2, \ldots, k\}$ and $\CE =  \{ (a,b) \in  \Gamma \times \Gamma \ | \  a \neq b \}$.

Any edge-checkable $\LCL$ problem $\Pi$ can be alternatively described in our formalism: 
$\Pi = (\Gamma,\CE,\CV,\CH,\CT)$, where
$\CV = \{ (a,a) \ | \ a \in \Gamma \}$ and
$\CH = \CT = \Gamma$. Hence $\Pi$ can also be seen as an  $\LCL$ problem on directed paths---in essence, $\Pi$ describes what are feasible label sequences when one follows any path from a leaf to the root. We claim that $\Pi$  has the same asymptotic round complexity on both rooted trees and directed paths, and hence our classification of $\LCL$ problems on directed paths also applies to edge-checkable $\LCL$ problems on rooted trees:

\begin{theorem}\label{thm:trees}
Let  $\Pi$ be any edge-checkable $\LCL$ problem on rooted trees. Then  $\Pi$ has the same asymptotic round complexity on both rooted trees and directed paths.
\end{theorem}

\begin{proof}
A directed path is a special case of a rooted tree, and hence lower bounds on directed paths automatically apply to rooted trees. The non-trivial part of the proof is to transform any given algorithm $\mathcal{A}$ for  $\Pi$ on directed paths to an algorithm $\mathcal{A}'$ for  $\Pi$ on rooted trees with the same asymptotic round complexity.

Let the node with only one tail port be the first node in a directed path.
We can assume that $\mathcal{A}$ is \emph{one-sided} in the sense that the output label of $v$ only depends on $v$ and the nodes that precedes $v$ in the directed path. To achieve that, we just need to shift the output labels by $T$ nodes, where $T$ is the runtime of  $\mathcal{A}$, and then assign the output labels of the first $T$ nodes of the directed path locally.

Given that $\mathcal{A}$ is one-sided,
the algorithm $\mathcal{A}'$ for  $\Pi$ on rooted trees applies $\mathcal{A}$ to each root-to-leaf path simultaneously. Because  $\mathcal{A}$ is one-sided, the output label of each node $v$ only depends on $v$ and its ancestors, and hence such a simultaneous execution is possible. The correctness of $\mathcal{A}'$ follows from the correctness of $\mathcal{A}$.
\end{proof}

\paragraph{Canonical algorithms for rooted trees.}

Recall that any $O(\log^\ast n)$-round $\LCL$ problem $\Pi$ on directed paths can be solved in a canonical way as follows. Let   $q$ be any flexible state with flexibility $k$. We find a distance-$k$ anchoring for the directed path, use $q$ at the anchor points, and fill the output labels in the gaps. 

The  proof of Theorem~\ref{thm:trees} implicitly implies that any $O(\log^\ast n)$-round edge-checkable $\LCL$ problem $\Pi$ on rooted trees can be solved in a canonical way analogously based on the following variant of distance-$k$ anchoring.

\begin{definition}\label{def:anchoring-tree}
A distance-$k$ anchoring for a rooted tree $T$  is a maximal subset of edges that splits the rooted tree into subtrees $T_1, T_2, \ldots, T_s$ satisfying the following conditions.
\begin{itemize}[noitemsep]
    \item The height of $T_i$ is at least $k-1$.
    \item If $v$ is a leaf in $T_i$ and a non-leaf in $T$, then the distance between the root of $T_i$ and $v$ equals the height of $T_i$.
\end{itemize}
\end{definition}

Such a distance-$k$ anchoring for a rooted tree $T$ can be computed in $O(\log^\ast n)$ rounds for $k = O(1)$ as follows. Apply any one-sided $O(\log^\ast n)$-round algorithm for computing a distance-$k$ anchoring for directed paths on each root-to-leaf path of $T$, shift down the output labels by $O(k)$ nodes, and finally split the large (possibly unbalanced) subtree near the root as appropriate.

\section{Discussion}\label{sec:conclusion}

We have seen that questions about the solvability of $\LCL$s in paths are closely related to classical automata-theoretic questions, as we can directly interpret a path as a string. Our work on $\LCL$s in cycles can be then seen as an extension of classical questions to \emph{cyclic words}. In particular, we see that an automaton ``accepts'' all but finitely many cyclic words if and only if there is a flexible state in the automaton, or equivalently if a D3-directing word exists for a strongly connected component of the automaton. Our work shows that all such questions on cyclic words can be decided in polynomial time, even if their classical non-cyclic analogs are in some cases $\coNP$-complete.

As we saw in Section~\ref{sec:trees}, our approach can be extended to the study of $\LCL$s beyond unlabeled paths and cycles. There are two main open questions after our work:
\begin{enumerate}
    \item What is the largest family of $\LCL$ problems in rooted trees for which round complexity can be decided in polynomial time? We now know that edge-checkable $\LCL$ problems can be characterized efficiently, while the general case is $\PSPACE$-hard.
    \item Is the round complexity of all $\LCL$ problems on rooted trees decidable? What about unrooted trees? For $\LCL$ problems on bounded-degree trees, we  know that it is decidable to distinguish between the complexity pairs $O(\log n)$ -- $n^{\Omega(1)}$  and $O(n^{1/(k+1)})$ -- $\Omega(n^{1/k})$ for any constant $k \geq 1$~\cite{Balliu2018disc,Chang2019,chang:LIPIcs:2020:13096}, but the general question for deciding the round complexity of $\LCL$ problems on trees is still widely open.
\end{enumerate}

\paragraph{Recent follow-up work.}
Subsequent to this work, the round complexity of $\LCL$ problems in the case of rooted \emph{regular} trees was studied in \cite{balliu21rooted-trees}, but the above questions still remain largely open. Polynomial-time decidability in rooted regular trees is still an open question. Rooted trees in general are not yet understood, and neither are unrooted regular trees.

\section*{Acknowledgments}

This work is an extended and revised version of a preliminary conference report that appeared in the 28th International Colloquium on Structural Information and Communication Complexity (SIROCCO 2021).

We would like to thank Alkida Balliu, Sebastian Brandt, Laurent Feuilloley, Juho Hirvonen, Yannic Maus, Dennis Olivetti, Aleksandr Tereshchenko, Jara Uitto, and all participants of the Helsinki February Workshop 2018 on Theory of Distributed Computing for discussions related to the decidability of $\LCL$s on trees. We would also like to thank the anonymous reviewers of previous versions of this works for their helpful comments and feedback.

Yi-Jun Chang was supported by Dr.~Max R\"{o}ssler, by the Walter Haefner Foundation, and by the ETH Z\"{u}rich Foundation.

{\DeclareUrlCommand\path{}
\def\UrlFont{\footnotesize\sf}
\bibliography{references}}

\end{document}